%% file: main.tex
\documentclass{article}

\usepackage{amsmath,amsthm,amssymb,mathtools}
\usepackage{hyperref}  
\hypersetup{colorlinks=true}
\hypersetup{linkcolor=[rgb]{.7,0,0}}
\hypersetup{citecolor=[rgb]{0,.7,0}}
\hypersetup{urlcolor=[rgb]{.7,0,.7}}

\usepackage{graphicx}
\usepackage[margin=1in]{geometry}
\usepackage{comment}
\usepackage{dsfont}
\usepackage{tikz}
\usetikzlibrary{calc}
\usepackage{algorithm}
\usepackage{paralist}
\usepackage[noend]{algpseudocode}

\usepackage{aliascnt} 

\newtheorem{theorem}{Theorem}[section]
\newtheorem*{theorem*}{Theorem}

\newaliascnt{definition}{theorem}
\newtheorem{definition}[definition]{Definition}
\aliascntresetthe{definition}

\newtheorem*{definition*}{Definition}

\newaliascnt{lemma}{theorem}
\newtheorem{lemma}[lemma]{Lemma}
\aliascntresetthe{lemma}

\newtheorem*{lemma*}{Lemma}

\newaliascnt{claim}{theorem}
\newtheorem{claim}[claim]{Claim}
\aliascntresetthe{claim}

\newtheorem*{claim*}{Claim}

\newaliascnt{fact}{theorem}

\aliascntresetthe{fact}

\newtheorem*{fact*}{Fact}

\newaliascnt{observation}{theorem}
\newtheorem{observation}[observation]{Observation}
\aliascntresetthe{observation}

\newtheorem*{observation*}{Observation}

\newaliascnt{conjecture}{theorem}

\aliascntresetthe{conjecture}

\newtheorem*{conjecture*}{Conjecture}

\newaliascnt{corollary}{theorem}
\newtheorem{corollary}[corollary]{Corollary}
\aliascntresetthe{corollary}

\newtheorem*{corollary*}{Corollary}

\newaliascnt{remark}{theorem}
\newtheorem{remark}[remark]{Remark}
\aliascntresetthe{remark}

\newtheorem*{remark*}{Remark}

\newaliascnt{proposition}{theorem}

\aliascntresetthe{proposition}

\newtheorem*{proposition*}{Proposition}

\newaliascnt{example}{theorem}

\aliascntresetthe{example}

\newtheorem*{example*}{Example}
\newcommand{\notshow}[1]{}

\usepackage{dsfont}

\DeclarePairedDelimiter{\card}{\lvert}{\rvert}
\def\alice{A}
\def\bob{B}
\def\comm{C_{\M}}
\def\shat{\mathsf{Sh}}
\def\nbr{\mathsf{nbr}}
\def\shatnbr{\shat\text{-}\nbr}

\usepackage{thm-restate}

\newcommand{\R}{\mathbb{R}}

\newcommand{\poly}{\operatorname{poly}}

\newcommand{\A}{\mathcal{A}}
\newcommand{\B}{\mathcal{B}}
\newcommand{\M}{\mathcal{M}}
\newcommand{\BigO}{\mathcal{O}}
\newcommand{\TypicalBob}{\T_{\bob}^{\mathsf{Typ}}}

\makeatletter
\newcommand\Ps@textstyle[2]{\mathbb{P}_{#1}\left[{#2}\right]}
\newcommand\Es@textstyle[2]{\mathbb{E}_{#1}\left[{#2}\right]}
\newcommand\Ps[2]{%
  \mathchoice 
  {\underset{{#1}}{\mathbb{P}}\left[{#2}\right]}
  {\Ps@textstyle{#1}{#2}}
  {\Ps@textstyle{#1}{#2}}
  {\Ps@textstyle{#1}{#2}}
}
\newcommand\Es[2]{%
  \mathchoice 
  {\underset{{#1}}{\mathbb{E}}\left[{#2}\right]}
  {\Es@textstyle{#1}{#2}}{\Es@textstyle{#1}{#2}}{\Es@textstyle{#1}{#2}}
}
\makeatother

\newcommand{\I}{\mathcal{I}}
\newcommand{\T}{\mathcal{T}}
\newcommand{\Mpar}{\M^{\mathrm{par}}}
\newcommand{\Mseq}{\M^{\mathrm{seq}}}

\title{Exponential Communication Separations between Notions of Selfishness}

\author{
  Aviad Rubinstein\thanks{
    Supported by NSF CCF- 1954927, and a David and Lucile Packard Fellowship.
    } \\
  Stanford University \\
  aviad@cs.stanford.edu
  \and Raghuvansh R. Saxena\thanks{
    Supported by the National Science Foundation CAREER award CCF-1750443 and a Microsoft PhD Fellowship.
    } \\
  Princeton University \\
  rrsaxena@cs.princeton.edu
  \and Clayton Thomas\thanks{
    Supported by NSF-CCF 1955205.}  \\
  Princeton University \\
  claytont@cs.princeton.edu
  \and S. Mathew Weinberg\footnotemark[3]  \\
  Princeton University \\
  smweinberg@princeton.edu
  \and Junyao Zhao\thanks{
    Supported by NSF CCF-1954927.
    }  \\
  Stanford University \\
  junyaoz@stanford.edu
}

\begin{document}

\maketitle

\begin{abstract}
We consider the problem of implementing a fixed social choice function between multiple players (which takes as input a type $t_i$ from each player $i$ and outputs an outcome $f(t_1,\ldots, t_n)$), in which each player must be incentivized to follow the protocol. In particular, we study the communication requirements of a protocol which: (a) implements $f$, (b) implements $f$ and computes payments that make it ex-post incentive compatible (EPIC) to follow the protocol, and (c) implements $f$ and computes payments in a way that makes it dominant-strategy incentive compatible (DSIC) to follow the protocol. 

We show exponential separations between all three of these quantities, already for just two players. That is, we first construct an $f$ such that $f$ can be implemented in communication $c$, but any EPIC implementation of $f$ (with any choice of payments) requires communication $\exp(c)$. This answers an open question of
[Fadel and Segal, 2009; Babaioff et. al., 2013].
Second, we construct an $f$ such that an EPIC protocol implements $f$ with communication $C$, but all DSIC implementations of $f$ require communication $\exp(C)$.
\end{abstract}



\clearpage
\input{Introduction.tex}

\input{Preliminaries.tex}

\input{PriceComputation.tex}

\input{DominantVsExPost.tex}

\input{DsicVsEpicFormalities.tex}

\input{DsicVsEpicWithTransfers.tex}


\bibliographystyle{alpha}
\bibliography{MasterBib}{}

\appendix

\input{FormalModel.tex}

\input{RemovingAssumptions.tex}

\end{document}

%% file: Introduction.tex
\section{Introduction}
Consider the following canonical problem: there is a set $Y$ of possible outcomes, and each of $n$ players have a type $t_i$ which determines their utility $u_i(t_i,y)$ for each outcome $y \in Y$. You have a particular \emph{social choice function} $f$ in mind, which maps a profile of types $\vec{t} = (t_1,\ldots,t_n)$ to $f(\vec{t}) \in Y$. A canonical question within Computer Science might first ask ``what is $CC(f)$, the communication complexity of $f$?'' That is, over all deterministic protocols computing $f$ among the $n$ players (who initially each know only their own type, and not that of any others), which one uses the least number of bits in the worst case?

But consider now the possibility that the players do not simply follow the intended protocol, and instead strive to maximize their own utility. The need to incentivize the players to follow the protocol motivates the entire field of Algorithmic Mechanism Design, as well as questions such as ``\emph{what is the communication complexity to implement $f$, using a protocol which incentivizes the players to follow it?}'' 

There are several formal instantiations of this question, depending on how strongly one wishes to incentivize the players. One common solution concept is ex-post incentive compatibility (EPIC), where the protocol may charge prices and it is in each player's interest to follow the protocol assuming that other players follow the protocol as well (formally, it is a Nash equilibrium to follow the protocol, no matter the other players' types). We let $CC^{EPIC}(f)$ denote the minimum communication cost of an EPIC protocol implementing $f$. Another common solution concept is dominant strategy incentive compatibility (DSIC), where the protocol may charge prices and it is in each player's interest to follow the protocol no matter what the other players do (even if that behavior is completely irrational). We let $CC^{DSIC}(f)$ denote the minimum communication cost of a DSIC protocol implementing $f$.
Because any EPIC protocol must compute $f$ and any DSIC protocol is in particular an EPIC protocol, we have $CC(f) \le CC^{EPIC}(f) \le CC^{DSIC}(f)$.

Formally, we study the following question: \emph{for a fixed $f$, how does $CC(f)$ relate to $CC^{EPIC}(f)$, and how does $CC^{EPIC}(f)$ relate to $CC^{DSIC}(f)$}?
While related directions have received substantial attention and produced a vast body of works (we overview this related work, and others, in \autoref{sec:related}), relatively little attention has been paid to these fundamental questions. Our main results are exponential separations between all three quantities (and these are the first such separations). Specifically:

\begin{theorem*}[See \autoref{thm:epic} and \autoref{thm:DsicVsEpic}] There exists $f$ such that $CC^{EPIC}(f) = \exp(CC(f))$. There exists $g$ such that $CC^{DSIC}(g) = \exp(CC^{EPIC}(g))$.
\end{theorem*}

The gap in both cases is at most exponential, so this is the largest gap possible.\footnote{To see this, consider the following sketch: for every protocol, there exists a simultaneous protocol (with one round of communication) with at most an exponential blowup in communication. A simultaneous protocol is EPIC if and only if it is DSIC, and~\cite[Proposition 1]{FadelS09} shows that EPIC prices can be added to any simultaneous protocol for $f$ with low overhead. So the gap between the three quantities can be no larger than the gap between simultaneous and interactive communication requirements for $f$, which is at most exponential.}

\subsection{Context and Related Work}\label{sec:related}
There is a \emph{vast} literature studying the communication requirements of protocols for honest players versus truthful mechanisms for strategic players~\cite{LehmannOS02,LaviMN03, LaviS05,NisanS06,DobzinskiS06, Dobzinski07, PapadimitriouSS08,Feige09, FeigeV10, DobzinskiNS10,DobzinskiN11, BuchfuhrerDFKMPSSU10, DobzinskiV13, DanielySS15, Dobzinski16b, BravermanMW18, AssadiS19, EzraFNTW19, AssadiKSW20}. Our work certainly fits into this literature, but goes in a fairly distinct direction. Specifically, this literature nearly-ubiquitously considers comparisons between how much communication is required for \emph{some} $f$ satisfying some property (e.g.~guaranteeing an $\alpha$-approximation to the optimal welfare\footnote{The welfare of an outcome $y$ is defined as $\sum_i u_i(t_i,y)$.}) versus how much communication is required for \emph{some} EPIC implementation of \emph{some} $g$ guaranteeing that property. In particular, $f$ and $g$ may be different social choice functions, and separations normally arise because the lowest-communication $f$ guaranteeing the desired property \emph{has no EPIC implementation} --- there simply don't exist prices that make any implementation of $f$ EPIC, no matter how much communication is used.\footnote{On the other hand, if $f$ is EPIC-implementable, it is DSIC-implementable, but perhaps with exponential overhead.} 

Our work studies a fundamentally different question: for a fixed $f$ \emph{which is EPIC-implementable}, how much communication overhead is required to actually compute prices which make the implementation EPIC? For an example of this distinction, consider a single-item auction: each player has a value $v_i$ for the item. The space of outcomes can award the item to any bidder, or no one. The social choice function $f$ which gives the item to the highest bidder can be EPIC-implemented (by the second-price auction). The social choice function $g$ which gives the item to the lowest bidder cannot be EPIC-implemented (by any prices, no matter how much communication). In general, many approximation algorithms for richer settings tend to be like $g$: they are simply not implementable, no matter what. So the driving force behind all prior work is separating the approximation guarantees for efficient protocols which are EPIC-implementable (and tend to have low overhead to actually compute the prices), versus those which are not. 

There is significantly less prior work addressing our specific questions. The direction was first posed in~\cite{FadelS09}, who explicitly pose the question of $CC(f)$ versus $CC^{EPIC}(f)$, and demonstrate that $CC^{EPIC}(f)$ can be strictly larger than $CC(f)$.  Follow-up work of~\cite{BabaioffBS13} were the first to make progress on this, and show a separation of $CC(f)$ versus $CC^{EPIC}(f)$ which is linear in the number of players (so in particular, the blow-up for two players is not large). In comparison to these works, our \autoref{thm:epic} shows the maximum possible gap (exponential) with just two players, resolving the open question in~\cite{FadelS09}.

\cite[Appendix B.2]{FadelS09} defines and discusses $CC^{DSIC}$, but only considers the relationship between $CC$ and $CC^{DSIC}$ (not the gap between $CC^{EPIC}$ and $CC^{DSIC}$). \cite[Appendix C.1]{Dobzinski16b} shows that no large separation between $CC^{EPIC}(f)$ and $CC^{DSIC}(f)$ is possible for the particular setting of two player combinatorial auctions with arbitrary monotone valuations\footnote{The proof of \cite[Appendix C.1]{Dobzinski16b} relies on the fact that incentive compatible combinatorial auctions with arbitrary monotone valuations have low ``taxation complexity''. Our construction in \autoref{sec:DominantVsExPost} circumvents this theorem because its environment is a very structured subset of two player monotone combinatorial auctions, and moreover, our social choice function $f$ has high taxation complexity. }.

The study of $CC^{EPIC}(f)$ versus $CC^{DSIC}(f)$ is conceptually related to a recent push with the Economics and Computation community to understand obviously strategyproof (OSP) mechanisms~\cite{Li17,BadeG17,AshlagiG18,PyciaT19}. These works do not focus on communication complexity, but rather on characterizing implementations which satisfy OSP (a stronger, but related, definition than DSIC). In comparison to these works, our \autoref{thm:DsicVsEpic} bears technical similarity, and our approach may be useful for proving communication lower bounds on OSP implementations.

\cite{FadelS09} also study related questions for a solution concept termed ``Bayesian incentive compatibility'' (BIC), and they obtain a tight exponential separation of $CC$ and $CC^{BIC}$. \cite{BabaioffKS15} studies the solution concept termed ``truthful in expectation'' (TIE), and show that in \emph{single-parameter} settings there is no (substantial) separation between $CC(f)$ and $CC^{TIE}(f)$.\\

\noindent\textbf{Concurrent and Independent Work.} Concurrently and independently of our work, Dobzinski and Ron~\cite{DobzinskiR21} also consider the relationship between $CC$ and $CC^{EPIC}$.\footnote{Both papers were uploaded to arXiv simultaneously on December 29th, 2020.} In particular, they also provide a construction of a function $f$ witnessing $CC^{EPIC}(f) = \exp(CC(f))$ (their Section~3.1), which is similar to ours (our \autoref{sec:CCvsEpic}) in that it derives hardness from high-precision prices. The remainder of their paper is disjoint from ours (in particular, they do not study $CC^{DSIC}$, so there is no analogue to our \autoref{sec:DominantVsExPost}). Instead, they establish the following results: (a) There exist functions $f$ with $CC^{EPIC}(f) = \exp(CC(f))$ \emph{without} high-precision prices (but with a third bidder). (b) Under certain assumptions on $f$, $CC^{EPIC}(f) = \poly(n, CC(f))$ and/or $CC^{TIE}(f) = \poly(n,CC(f))$. (c) Reconstructing the \emph{menu} presented by an EPIC mechanism can be exponentially harder than computing the mechanism alone. A high-level distinction of our works is that our paper provides exponential separations between multiple solution concepts (algorithmic vs. EPIC vs. DSIC), whereas their paper provides a more thorough investigation of algorithmic vs. EPIC. 

\subsection{Summary and Roadmap}
We establish an exponential separation between $CC(f)$ and $CC^{EPIC}(f)$, and also $CC^{EPIC}(f)$ and $CC^{DSIC}(f)$, both the largest possible, and first of their kind. \autoref{sec:CCvsEpic} provides the separation between $CC(f)$ and $CC^{EPIC}(f)$, and \autoref{sec:DominantVsExPost} provides the separation between $CC^{EPIC}(f)$ and $CC^{DSIC}(f)$.

%% file: Preliminaries.tex
\section{Preliminaries}


We study \emph{implementations} of \emph{social choice functions} over
\emph{(social choice) environments}.
For completeness and accessibility for the reader not familiar with game
theory, we rigorously define all of these terms in~\autoref{app:FormalModel}.
Here, we briefly and intuitively describe the central definitions of the paper.

The environment specifies a set of outcomes $Y$ and
a set of types $\T_1,\ldots,\T_n$ for the $n$ different strategic agents.
Intuitively, the types represent the different possible options for ``who
each agent might be''.
When agent $i$ has type $t_i\in\T_i$, they
have utility $u_i(t_i, y)\in\R$ for each outcome $y\in Y$.
When we study environment with transfers\footnote{
  Throughout the paper, we make no assumptions on the transfers.
  That is, they can be positive or negative,
  and an agent can receive negative utility.
  This makes our impossibility results only stronger.
}, we assume utilities are
quasilinear (that is, if outcome $y$ is selected and agent $i$
receives transfer $p$, then agent $i$ gets utility $u_i(t_i, y)+p$).
The social choice function $f : \T_1\times\ldots\times\T_n\to Y$
specifies how the outcome depends on the type each agent has.
While the ``social planner'' designing the mechanism
wishes to compute $f$, the agents wish to maximize their own utility.
The social choice function itself is assumed to be implementable. That is, there exists transfer functions
$p_1,\ldots,p_n : \T_1\times\dots\times\T_n \to \R$ for each agent, such that
for all $i$, types $t_1,\ldots,t_n$, and $t_i'$, we have
\[ f(t_i, t_{-i}) + p_i(t_i, t_{-i})
  \ge f(t_i', t_{-i}) + p_i(t_i', t_{-i}).
\]
We say that transfers $(p_1,\ldots,p_n)$ \emph{incentivize} $f$,
and we say that $f$ is incentive compatible without transfers if each
$p_i(\cdot)$ above can be taken to be $0$.

A mechanism consists of an (extensive form) game $G$ which the $n$ agents
play, and ``type-strategies'' $S_1,\ldots,S_n$ which suggest how the agents
should play $G$.
Intuitively, the game $G$ iteratively solicits actions from players,
updating its state according to the action chosen,
and outputting some result after a finite amount of time.
This is represented by a game tree, where the nodes correspond to states of
the game. Each non-leaf node is labeled by some agent,
and the edges below that node are labeled with the actions that agent may
play at that state of the game.
The game is not perfect information: it may hide
information from agents or ask them to act simultaneously.
For each agent $i\in[n]$, the states of $G$ at which $i$ is called to act
are partitioned into ``information sets'' $I_i \in \I_i$,
where two nodes are in the same information set if and only if agent $i$
cannot distinguish between them while playing the game\footnote{
  We assume the game satisfies ``perfect recall'', that is, the game cannot
  force agents to forget information they knew in the past.
  For details on how information sets are defined, see
  \autoref{app:FormalModel}.
}.
For $i\in[n]$, the type-strategy $S_i$ maps types $t_i \in \T_i$ to
``behavioural strategies'' $s_i = S_i(t_i)$ which player $i$ can play in $G$.
A behavioural strategy (typically referred to simply as a strategy)
specifies the action that player $i$ will choose
any time they are called to act over the course of the game,
that is, it assigns an action to each information set of player $i$.
We denote the result output by $G$ when the agents play strategies
$s_1,\ldots,s_n$ by $G(s_1,\ldots,s_n)$.

A mechanism $G$ with strategies $S_1,\ldots, S_n$ \emph{computes (without transfers)} a social choice
function $f$ if $G(S_1(t_1),\ldots,S_n(t_n))\allowbreak = f(t_1,
\allowbreak \ldots,t_n)$.
A mechanism \emph{computes} $f$ \emph{(with transfers)} if the result of
the game additionally includes transfers $p_1,\ldots,p_n$ to each
player.
\notshow{\footnote{
  As is often the case in the literature (for example,
  in welfare-maximization context),
  we treat transfers as simply a way to incentivize agents
  to tell the truth. So our considerations do not fit cases where,
  for example, the mechanism must be ``budget balanced''
  or the social planner wishes to maximize profit.
  Additionally, in our complexity measures, we reason about the
  ``easiest to compute'' price function which incentivize
  the social choice function. Note that this makes our separations only
  stronger.
}.}

We consider two notions of incentive compatibility for interactive
mechanisms. In words, a mechanism is dominant strategy
incentive compatible (DSIC)
if, for \emph{any} (behavioral) strategy profile $s_{-i} := (s_j)_{j\ne i}$
of the other players, it is a best response for player $i$ to play $S_i(t_i)$.
That is, for all $t_i, s_{-i}, s_i'$, we have
\[ u_i( t_i, G( S_i(t_i), s_{-i} ) ) 
  \ge u_i( t_i, G( s_i', s_{-i} ) ),  
\]
where we recall that if the mechanism has transfers, $u_i(t_i,\cdot)$
is the quasilinear utility given by agent $i$'s value for the outcome when
their type is $t_i$, plus the transfer $p_i$ to player $i$.
On the other hand, a mechanism is ex-post Nash incentive compatible (EPIC)
if, for any profile of strategy $S_{-i}(t_{-i}) := ( S_j(t_j) )_{j\ne i}$ 
which are \emph{consistent with type-strategies $S_{-i}$},
it is a best response to play $S_i(t_i)$.
That is, for all $t_i, t_{-i}, s_i'$, we have
\[ u_i( t_i, G( S_i(t_i), S_{-i}(t_{-i} ) ) 
    \ge u_i( t_i, G( s_i', S_{-i}(t_{-i}) ) ). 
\]

Observe quickly the following approach for an EPIC implementation of $f$: Say that $(p_1,\ldots, p_n)$ incentivizes $f$. Then one can run protocols separately to compute $f$, and also to compute each $p_i$, and then output all of these together. This is simply because the EPIC constraints assume that the other bidders' strategies are fixed by their type. So the overhead of $CC^{EPIC}(f)$ versus $CC(f)$ is exactly the overhead to compute transfers. This \emph{does not} hold for DSIC implementations. Indeed, this is because other bidders may use a bizarre (not utility-maximizing) strategy which changes their behavior in (e.g.) the protocol to compute $p_i$ as a function of your behavior in the protocol to compute $f$. But the EPIC condition does not require guarantees against such bizarre strategies, only the fixed strategies which guarantee each player a best response (assuming other players also use such a strategy). We formally define our complexity measures as follows:
\begin{definition}
  For an arbitrary social choice function $f$,
\begin{itemize}
  \item $CC(f)$ is the minimum communication cost of a mechanism
    (no incentives) computing $f$.
  \item If $f$ is implementable, $CC^{EPIC}(f)$ is the minimum value of
    $CC(f, p_1,\ldots,p_n)$ over any transfer functions
    $p_1,\ldots,p_n$ which incentivize $f$.
  \item If $f$ is implementable,
    $CC^{DSIC}(f, p_1,\ldots,p_n)$ is the minimum communication cost of
    any DSIC mechanism computing $(f, \allowbreak p_1, \allowbreak \ldots,p_n)$.
    Moreover, $CC^{DSIC}(f)$ is the minimum value of
    $CC^{DSIC}(f, p_1,\ldots,p_n)$ for any transfer functions
    $p_1,\ldots,p_n$ which incentivize $f$.
\end{itemize}
\end{definition}

%% file: PriceComputation.tex

\section{
  Exponential Separation of 
  \texorpdfstring{$CC(f)$}{CC(f)} and 
  \texorpdfstring{$CC^{EPIC}(f)$}{CC[EPIC](f)}
}
\label{sec:CCvsEpic}

In this section, we show that there exists an implementable
social choice function $f$
which has communication complexity $\BigO(\log n)$, yet any EPIC
implementation of $f$ must use $\Omega(n)$ communication. 

We now describe our construction at a high level.
Our instance has two players, Alice and Bob.
Alice's type can be represented succinctly, but Bob's
type is ``complicated''. 
Therefore, without regards to incentives,
this social choice function can be efficiently computed
in two rounds, with Alice sending her type to Bob in the
first round, and Bob deciding the outcome in the second round.
However, the social choice function and the utilities of Alice
are designed carefully such that there is essentially only one
possible transfer function that gives an EPIC implementation,
and moreover, this transfer function has to be
``as complicated as the types of Bob''.
This means that the communication required to EPIC
implement the social choice function is large.

\paragraph{Social choice environment.} Consider a 2-player social choice
environment and refer to the players 
as Alice and Bob. The space of outcomes of the environment is $[n+1]$. The
class of Bob's types is $\T_B=\{0,1\}^n$. That is, Bob's
type is a binary string $b$ of length $n$. 
We let $b_i\in\{0,1\}$ denote $b$'s $i$-th coordinate.
Bob's utility is always zero regardless of the
outcome (that is, $u_{\bob}(b,i)=0$ for all $i\in[n+1], b\in\T_B$).
The class of Alice's types is 
$\T_A=\bigcup_{i\in[n]}\{a_{i, \ell}, a_{i, \ell}', a_{i, h}, a_{i, h}'\}$,
where for each $i\in[n]$, the types
$a_{i, \ell}, a_{i, \ell}', a_{i, h}, a_{i, h}'$ have 
utility:
\begin{align*}
  & u_A(a_{i, \ell}, i) = 2^{-n} && u_A(a_{i, \ell}', i)= 0 
    && \\ & u_A(a_{i, \ell}, i+1) = 0 && u_A(a_{i, \ell}', i+1) = 2^{-n} \\
  & u_A(a_{i, h}, i) = 2^{-n} && u_A(a_{i, h}', i) = 0
    && \\ & u_A(a_{i, h}, i+1) = 2^{i} && u_A(a_{i, h}', i+1) = 2^{i}+2^{-n},
\end{align*}
and $u_A(a_{i, \ell}, j) = u_A(a_{i, \ell}', j) = u_A(a_{i, h}, j) =
u_A(a_{i, h}', j) = - \infty$ for all other outcomes $j\notin\{i,i+1\}$.
Intuitively, $a_{i,\ell}, a'_{i,\ell}$ are ``low types'' of Alice, and
$a_{i,h},a'_{i,h}$ are ``high types'' (which get much more utility from
outcome $i+1$).

\paragraph{Social choice function.}
The social choice function $f:\T_A\times\T_B\to [n+1]$ is given by
\begin{align*}
    & f(a_{i, \ell}, b) = i 
    && f(a_{i, \ell}', b) = i+1-b_i \\
  & f(a_{i, h}, b) = i+1-b_i
    &&  f(a_{i, h}', b) = i+1.
\end{align*}
That is, each of Alice's type among $a_{i, \ell}, a_{i, \ell}', a_{i, h},
a_{i, h}'$ receives either outcome $i$ or $i+1$, and the exact outcome chosen
depend on Bob's type $b$ in the following way:
If $b_i = 0$, then $a_{i, \ell}$ receives outcome $i$,
and each of $a_{i, \ell}', a_{i, h}, a_{i, h}'$ receives outcome $i+1$.
If $b_i = 1$, then each of $a_{i, \ell}, a_{i, \ell}', a_{i, h}$ receives outcome $i$,
and $a_{i, h}'$ receives outcome $i+1$.


\begin{theorem}\label{thm:epic}
  In the 2-player environment above,
  the social choice function $f$ is EPIC implementable.
  Moreover, there is an exponential separation between the communication
  complexity for computing $f$ and the communication complexity of any EPIC
  implementation of $f$, i.e.,
  \[ CC(f) = \BigO(\log n)
    \qquad \qquad
    CC^{EPIC}(f) = \Theta(n).
  \]
\end{theorem}
\begin{proof}
First, observe that Alice and Bob can compute $f$ with $\BigO(\log n)$
communication in the following way: Alice sends her valuation, which can
be described with $\BigO(\log n)$ bits, to Bob, and then, Bob computes and
outputs the outcome, which also costs $\BigO(\log n)$ bits.
Thus, $CC(f) = \BigO(\log n)$.

On the other hand, consider any EPIC implementation of $f$. 
Without loss of generality, we may assume that the transfers to Bob are always
$0$. Let $p(a,b)$ denote the transfer given to Alice when Alice has type
$a\in\T_A$ and Bob has type $b\in\T_B$.
By standard arguments, we must have $p(a',b)=p(a,b)$ for any $b\in\T_B$
and $a, a'\in \T_A$ such that $f(a',b)=f(a,b)$ 
(otherwise, one of $a$ or $a'$ would want to deviate to the other,
in order to get a higher transfer for the same outcome).
Thus, going forward we write the transfer function $p : [n+1]\times \T_B \to \R$,
where $p(i,b)$ is the transfer to Alice when Bob has type $b$
and outcome $i$ is the output of $f$.



Now we prove our main lemma, which allows us to characterize 
$p$ in any EPIC implementation of $f$.

\begin{lemma}
  \label{lem:keyDifference}
  Transfers $p$ incentivize $f$ if and only if we have 
  \[ p(i,b) - p(i+1,b) \in [b_i 2^i-2^{-n},\,b_i 2^i+2^{-n}] \tag{*}\label{eqn:keyDifference} \]
  for all $i\in[n]$ and $b\in\T_B$.
\end{lemma}
\begin{proof}
  When Alice has type
  $a_i \in \{ a_{i, \ell}, a_{i, \ell}', a_{i, h}, a_{i, h}' \}$,
  the social choice function $f$ will select outcome $i$ or $i+1$,
  based on the type of Alice and bit $b_i$ of Bob's valuation $b\in\T_B$.
  Certainly Alice will not want to deviate to an outcome $j\notin \{i,i+1\}$,
  as her utility for these outcomes is $-\infty$.
  Thus, to prove the ``if'' direction, 
  it suffices to show that for each $i\in[n]$ and $b\in\T_B$,
  when transfers satisfy (\ref{eqn:keyDifference})
  for this value of $i$ and $b$, if Alice has a type 
  $a_i \in \{ a_{i, \ell}, a_{i, \ell}', a_{i, h}, a_{i, h}' \}$, 
  she will not want to deviate to
  the unique outcome in $\{i,i+1\}\setminus \{f(a_i,b)\}$.
  To prove the ``only if'' direction, it suffices to show that if transfers $p$
  incentivize $f$, then (\ref{eqn:keyDifference}) must hold for each
  $i\in[n]$ and $b\in\T_B$.
  To this end, consider any $i\in [n]$.

  First, suppose $b_i=0$.
  This means that $a_{i,\ell}$ receives $i$,
  and $a_{i,\ell}',\allowbreak a_{i,h},a_{i,h}'$ receive $i+1$.

  Suppose that transfers $p$ satisfy (\ref{eqn:keyDifference}),
  i.e. $p(i,b)-p(i+1,b)\in[-2^{-n},\,2^{-n}]$.
  First, note that $a_{i,h}$ and $a_{i,h}'$ will not want to deviate to $i$,
  because these types have much higher utility for $i+1$ (and receive almost the same
  transfer on these two outcomes).
  Second, note that $u_A(a_{i, \ell}, i+1)-u_A(a_{i, \ell}, i)=-2^{-n}$ and $u_A(a_{i, \ell}', i+1)-u_A(a_{i, \ell}', i)=2^{-n}$, and it follows by $p(i,b)-p(i+1,b)\in[-2^{-n},\,2^{-n}]$ that
  \begin{align*}
      u_A(a_{i, \ell}, i)+p(i,b)\ge u_A(a_{i, \ell}, i+1)+p(i+1,b) \\
      u_A(a_{i, \ell}', i+1)+p(i+1,b)\ge u_A(a_{i, \ell}', i)+p(i,b)
  \end{align*}
  Thus, $a_{i,\ell}$ and $a_{i,\ell}'$ will not want to deviate either.

  Now we show that if transfers $p$ incentivize $f$, then they must satisfy
  (\ref{eqn:keyDifference}) for this value of $i$.
  Observe that $a_{i,\ell}$ and $a_{i,\ell}'$ have almost the same utility for $i$
  and $i+1$, yet receive different outcomes.
  This will force $p(i,b)-p(i+1,b)\in[-2^{-n},\,2^{-n}]$.
  Specifically, for neither of $a_{i,\ell}$ nor $a_{i,\ell}'$ to want to deviate
  to each other, we must have
  \begin{align*}
    & 2^{-n}+p(i,b) = u_A(a_{i, \ell}, i)+p(i,b)
      \\ & \qquad \ge u_A(a_{i, \ell}, i+1)+p(i+1,b) = p(i+1,b) \\
    & 2^{-n}+p(i+1,b) = u_A(a_{i, \ell}', i+1)+p(i+1,b)
      \\ & \qquad \ge u_A(a_{i, \ell}', i)+p(i,b) = p(i,b) ,
  \end{align*}
  and thus $p(i,b)-p(i+1,b)\in[-2^{-n},\,2^{-n}]$.

  Second, suppose $b_i=1$.
  This means that $a_{i,\ell},a_{i,\ell}',a_{i,h}$ receive $i$,
  and $a_{i,h}'$ receives $i+1$.
  The logic in this case is analogous to the first case.

  Suppose that transfers $p$ satisfy (\ref{eqn:keyDifference}),
  i.e. $p(i,b)-p(i+1,b)\in[2^i-2^{-n},\,2^i+2^{-n}]$.
  First, note that $a_{i,\ell}$ and $a_{i,\ell}'$ will not want to deviate to $i+1$,
  because these types have almost the same utility for $i$ and $i+1$ 
  (and receive a much higher transfer on $i$).
  Second, note that $u_A(a_{i, h}, i+1)-u_A(a_{i, h}, i)=2^i-2^{-n}$ and $u_A(a_{i, h}', i+1)-u_A(a_{i, h}', i)=2^i+2^{-n}$, and it follows by $p(i,b)-p(i+1,b)\in[2^i-2^{-n},\,2^i+2^{-n}]$ that
  \begin{align*}
      u_A(a_{i, h}, i)+p(i,b)\ge u_A(a_{i, h}, i+1)+p(i+1,b) \\
      u_A(a_{i, h}', i+1)+p(i+1,b)\ge u_A(a_{i, h}', i)+p(i,b)
  \end{align*}
  Thus, $a_{i,h}$ and $a_{i,h}'$ will not want to deviate either.

  Now we show that if transfers $p$ incentivize $f$, then they must satisfy
  (\ref{eqn:keyDifference}).
  Observe that $a_{i,h}$ and $a_{i,h}'$ have almost the same utilities for $i$ and $i+1$,
  yet receive different outcomes.
  This will force $p(i,b) - p(i+1,b)\in[2^i-2^{-n},\,2^i+2^{-n}]$.
  Specifically, for neither of $a_{i,h}$ nor $a_{i,h}'$ to want to deviate
  to each other, we must have
  \begin{align*}
    & 2^{-n}+p(i,b) = u_A(a_{i, h}, i)+p(i,b) \\
    & \qquad  \ge u_A(a_{i, h}, i+1)+p(i+1,b) = 2^i + p(i+1,b) \\
    & 2^i+2^{-n} + p(i+1,b) = u_A(a_{i, h}', i+1)+p(i+1,b) \\
    &  \qquad \ge u_A(a_{i, h}', i)+p(i,b) = p(i,b) ,
  \end{align*}
  and thus $p(i,b) - p(i+1,b)\in[2^i-2^{-n},\,2^i+2^{-n}]$.

\end{proof}

We now define transfers $p^*$ such that
\[  p^*(i,b) = - \sum_{j=1}^{i-1} b_j 2^j. \]
For each $b\in\T_B$ and $i\in[n]$, we have $p^*(i,b)-p^*(i+1,b)\in[b_i2^i-2^{-n},\,b_i2^i+2^{-n}]$,
and thus by \autoref{lem:keyDifference}, these transfers incentivize $f$.
Thus, let $\M$ denote the mechanism which has Alice announce her type
(using $\BigO(\log(n))$ bits),
tells that type to Bob, and then has Bob decide the outcome $i$
(using $\BigO(\log(n))$ bits) and the transfer $p^*(i,b)$ for Alice
(using $\BigO(n)$ bits).
This mechanism EPIC implements $f$ with communication cost $\BigO(n)$.


On the other hand, consider any mechanism $\M$ which EPIC implements $f$.
Let $p$ denote the transfers $\M$ gives to Alice.
By \autoref{lem:keyDifference} and telescoping sum, 
the transfers must satisfy $p(1,b)-p(n+1,b)\in[\sum_{j=1}^n (b_j2^j-2^{-n}),\,\sum_{j=1}^n (b_j2^j+2^{-n})]$ for
all $b\in\T_B$. Notice that for sufficiently large $n$, $n2^{-n}$ is tiny, and hence, the intervals $[\sum_{j=1}^n (b_j2^j-2^{-n}),\,\sum_{j=1}^n (b_j2^j+2^{-n})]$ corresponding to distinct $b$'s are disjoint.
Since there are $2^n$ distinct $b$'s, there are also $2^n$
distinct values of $p(1,b)-p(n+1,b)$.
Suppose for contradiction that $\M$ computes $p$ using ${o}(n)$ bits of
communication.
Then there also exists a protocol which can compute $p(1,b)-p(n+1,b)$ 
with ${o}(n)$ communication, which is impossible because there are $2^n$ such
values.
Therefore, any EPIC implementation of $f$ must have
communication cost $\Omega(n)$.
This completes the proof.

\end{proof}

\paragraph{Discussion.}

In the proof above, we showed that computing the
transfers requires large amount of communication because the transfers require
a large number of bits to represent. For two players, this is necessary.
That is, in a two player environment, if a social choice function $f$
can be incentivized with transfers that can be represented with $K$ bits,
then there exists an EPIC implementation with communication cost $CC(f)+K$.
This implementation first has
Alice and Bob compute the social choice function using an optimal protocol,
which requires $CC(f)$ bits, and then has each player
specify the transfer for the other player (as we recalled in the proof of
\autoref{thm:epic}, the transfers to Alice are determined solely by the outcome 
and Bob's type and vice versa).

We note that it is possible to modify the environment by giving Bob
nontrivial utilities such that $f$ is the unique social choice function
which maximizes the welfare $u_A(a,i)+u_B(b,i)$.
Specifically, for each Bob
type $b\in\T_B$, we define Bob's utility as
$u_{B}(b,i)=-\sum_{j=1}^{i-1}b_j2^j$ for each outcome $i\in[n+1]$, which is equal to $p^*(i,b)$ in the proof. 
In this modified environment, $f$ always returns the unique outcome 
which maximizes welfare. 
Notice that $p^*(i,b)$ then becomes the VCG transfer (up to an additive
constant that can depend Bob's type) for Alice. If we also let Alice output
the VCG transfer (up to an additive constant that can depend on Alice's
type) $p'(i,a):=u_A(a,i)$ for Bob after the outcome is decided, then
$p'$ along with $f$ is EPIC for Bob.  Together, $p^*,p'$ give an EPIC
implementation of $f$. 

Finally, in the above modified environment where $f$ is welfare-maximizing,
note that despite Alice's valuation being succinctly representable,
her utilities are ``high precision''. 
This is necessary, because by~\cite[Proposition 2]{FadelS09}, 
if all the valuations in the environment have low precision, every
welfare-maximizing social choice function has an EPIC implementation with only
slightly more communication for computing the transfers. Moreover, Bob's type
requires many bits to represent. This is also necessary,
because if both players have succinct types, they can simultaneously output their types,
after which the mechanism computes the correct outcome and charges VCG transfers.


%% file: DominantVsExPost.tex

\section{
  Exponential Separation of 
  \texorpdfstring{$CC^{EPIC}(f)$}{CC[EPIC](f)} and 
  \texorpdfstring{$CC^{DSIC}(f)$}{CC[DSIC](f)}
}
\label{sec:DominantVsExPost}

In this section, we construct a social choice function $f$ such that
$CC^{EPIC}(f) = \BigO(n)$, yet $CC^{DSIC}(f) = \exp(n)$.

\subsection{Building Up to Our Construction}

We walk through a list of examples of environments and social choice rules,
trying to build to an exponential separation of the communication required to
EPIC implement and DSIC implement the rules.
The first example is a classical illustration of the difference between ex-post
and dominant strategy implementations for extensive form games.

\subsubsection{Attempt One}
\label{sec:attempt1}


Consider a second price auction with two bidders, Alice and Bob, and a
single item, such that Alice's and Bob's value for the item are integers in $\{1,2,\ldots,10\}$.
If the auction is implemented as a direct revelation mechanism, then it is DSIC.
However, suppose we first ask Alice her value, then tell that value to Bob
and ask him to respond with his own value. This mechanism is no longer DSIC.
For example, one strategy of Bob is to always say his
value is $1$, except when Alice bids $8$, in which case he will say his
value is $9$. When Bob plays this strategy and Alice's true value is $8$,
Alice gets more utility by lying and bidding $9$ than by telling the truth.

We note that the above strategy for Bob is ``crazy'' in the sense that it
does not maximize his own utility, but serves mostly to
incentivize non-truthful bidding by Alice. Moreover, this crazy strategy
for Bob was possible only because Bob knew Alice's value and decided his
response as a function of this value. Observe that, for such a crazy
strategy to work, Bob does not have to know Alice's value exactly.
Intuitively and informally, the following two conditions suffice:
\begin{enumerate}[(a)]
\item \label{item:crazy1} Bob learns information about Alice's type.
\item \label{item:crazy2} Bob has two possible responses,
  one which gives Alice high utility, and one which give Alice low utility.
\end{enumerate}
 
Our next idea is to construct an instance where any low communication
mechanism must satisfy \autoref{item:crazy1} and \autoref{item:crazy2}
above. We first focus on \autoref{item:crazy1} and try to devise an
instance where any low-communication mechanism requires Bob to know
something about Alice's valuation. For this, we embed the
well-known ``Index'' problem from communication complexity in a
welfare-maximization context. Recall that, in the Index problem, there is
a parameter $K > 0$ such that Alice has an index $k \in [K]$ and Bob has a
vector $X = (x_i)_{i \in [K]} \in \{0,1\}^K$, and the goal is to output
the $k^{\text{th}}$ location in the vector $X$, {i.e.} $x_k$.

Intuitively, the importance of the Index problem lies in the fact the
only way to efficiently solve this problem is for Alice to reveal a lot of
information about her input.
Specifically, first observe that
the protocol where Alice sends $k$ to Bob, and Bob then simply
outputs $x_k$, uses communication $\BigO(\log K)$.
However, it turns out that any protocol that does not
reveal a lot of information about Alice's input
to Bob must have communication $\Omega(K)$ 
(this can be formalized, see \cite[{\em etc.}]{KushilevitzN97},
although we do not need to do so here).


\subsubsection{Attempt Two}
  \label{sec:attempt2}

Consider an auction where there are two bidders and an even number $m$ of
items for sale. The bidders, Alice and Bob, are
multi-minded\footnote{
  Recall that a valuation function $v$ on $[m]$ is multi-minded if there
  exists a collection $\{(v_i, T_i)\}_i$, where each $v_i\in\R$
  and $T_i\subseteq [m]$, such that
  $v(S) = \max \{ v_i | T_i \subseteq S \}$.
  The sets $T_i$ are call the ``interests'' of the valuation function $v$.
} with interests as follows: Alice is interested in exactly two sets,
a set $S \subseteq [m]$ of size $m/2$ that she values at $4$, and the set
$\overline{S}$ that she values at $1$. Bob's valuation is such that for
every subset $T \subseteq [m]$ of size $m/2$, he is interested in exactly
one of the sets $T$ and $\overline{T}$, which he values at $5$
(and he values the other set at $0$).
The social
choice function $f$ outputs the welfare-maximizing allocation of items
between Alice and Bob. That is, Bob gets
whichever of $S$ or $\overline{S}$ he values at $5$, and Alice gets the
complement (which she values at either $4$ or $1$).
Observe that $f$ is incentive compatible without transfers.


The direct revelation mechanism $\M_1$ 
(where Alice and Bob simultaneously reveal their entire type) is DSIC.
In this mechanism, Bob does not learn anything about Alice's type,
that is, \autoref{item:crazy1} in \autoref{sec:attempt1} does not hold. 
However, the fact that Bob communicates his entire type means that 
$\M_1$ requires communication exponential in $m$.


There is also a mechanism $\M_2$ for the above instance where
\autoref{item:crazy1} is satisfied. This is the mechanism that first asks
Alice for the set $S$ of size $m/2$ she values at $4$, and then asks Bob
which of the sets $S$ and $\overline{S}$ he values at $5$. The mechanism
$\M_2$ then gives Bob the set he said he values at $5$ and gives Alice the
complement. Observe that $\M_2$ is EPIC and requires $\mathcal{O}(m)$
communication. However, the mechanism $\M_2$ is not DSIC. Indeed, consider
a (crazy) strategy for Bob where he always says that the set $S$ reported
by Alice is the one he values at $5$ (regardless of his input). With this
strategy for Bob, Alice always gets the complement of what she reports and
therefore, she is incentivized to lie and report the set $\overline{S}$
instead of the set $S$ which is truly her favorite.



\paragraph{A low communication DSIC mechanism.}

However, the instance above does not yield a separation between the
communication complexity of DSIC and EPIC implementations, as there is an
$\mathcal{O}(m)$-communication mechanism that is also DSIC.
This mechanism, which we we call $\M^{\star}$,
asks Alice only report the sets $\{S, \overline{S}\}$
of size $m/2$ she has non-zero value for, without
specifying which one of the two she values at $4$. Then, the mechanism
$\M^{\star}$ asks Bob which of the sets $S$ and $\overline{S}$ he values at
$5$, gives him that set and gives Alice the complement of the set.

The mechanism $\M^{\star}$ clearly has communication $\mathcal{O}(m)$.
It is DSIC, because if Alice reports anything other than than 
$\{S, \overline{S}\}$, she will get utility $0$ regardless of what Bob says.
In particular, it is not possible
to construct a ``crazy'' strategy of Bob as in $\M_2$,
because Bob's response cannot depend on the difference between $S$
and $\overline{S}$.

In other words, the reason the mechanism $\M^{\star}$ is DSIC is that it
does not satisfy \autoref{item:crazy2} above. 
Even though Bob 
learns a lot of information about Alice's type,
he cannot respond to this information in a way that gives Alice a lower
utility in some cases, and a higher utility in other cases.


\paragraph{Need for new ideas.}
It may seem at first that the mechanism $\M^{\star}$ works only because in
our instance, Bob does not need to which of $S$ and $\overline{S}$ does
Alice value at $4$ in order to determine the welfare-maximizing allocation.
However, this is not the case. 
Even if the welfare-maximizing allocation was
dependent on which of $S$ and $\overline{S}$ is valued at $4$ by Alice, Bob
could just send two answers, one for the case when $S$ is valued at $4$ and
the other one for the when $\overline{S}$ is valued at $4$.
The resulting mechanism would still be DSIC.
Thus, new ideas are needed to get a
separation between the communication complexity of EPIC and DSIC
implementations.

\subsection{Construction and Intuition}
\label{sec:construction}

At a high level, our main construction is simply two
independent copies of the instance described in \autoref{sec:attempt2},
where the valuation
functions for Alice and Bob are additive over the two copies.

Formally, for every even $m$, we have a two player combinatorial auction
where a set $M_1 \sqcup M_2$ of items satisfying $\card*{M_1} = \card*{M_2}
= m$ is for sale. 
The set of outcomes is defined by\footnote{
  We restrict the auction to always award half of the items in $M_i$ to each
  bidder, for each $i\in[2]$.
  This restriction is without loss of generality,
  because the social choice function $f$
  always outputs allocations with this property,
  but it simplifies the notation slightly.
}
\[
  Y = \{ (X_1, X_2) \mid X_1\subseteq M_1,
  \ X_2\subseteq M_2,\ \card*{X_1} = \card*{X_2} = m/2\} .
\]
An outcome $(X_1, X_2)$ indicates that Alice receives $(X_1, X_2)$
and Bob receives $(\overline{X_1}, \overline{X_2})$.
Alice's types are also given by the set $\T_{\alice} = Y$ and her utility
function $u_{\alice} : \T_{\alice} \times Y \to \mathbb{R}$ is defined by
$u_{\alice}((S_1, S_2),\allowbreak (X_1, X_2))\allowbreak = u_{\alice, 1}(S_1, X_1) 
+ u_{\alice, 2}(S_2, X_2)$, where, for $i \in [2]$, we have:
\[
 u_{\alice, i}(S_i, X_i) = \begin{cases}
 4, &\text{~if~} S_i = X_i \\ 
 1, &\text{~if~} S_i = \overline{X_i} \\
 0, &\text{~otherwise}.
 \end{cases}
\]
Bob's type set $\T_{\bob}$ is the collection of all pairs
$(v_{\bob, 1}, v_{\bob, 2})$, where for $i \in [2]$, the function
$v_{\bob, i}$ maps a subset of $M_i$ of size $m/2$ to the set $\{0,5\}$
such that for each set $T \subseteq M_i$, $\card*{T} = m/2$, we have
$v_{\bob, i}(T) = 5$ and $v_{\bob, i}(\overline{T}) = 0$ or vice-versa.
Bob's utility function is:
\[
u_{\bob}((v_{\bob, 1}, v_{\bob, 2}), (X_1, X_2))
  = v_{\bob, 1}(\overline{X_1}) + v_{\bob, 2}(\overline{X_2}) .
\] 

Finally, the goal of the auctioneer is to maximize the welfare. Observe
that, if Alice's type is $(S_1, S_2) \in \T_{\alice}$ 
and Bob's type is $(v_{\bob, 1}, v_{\bob, 2}) \in \T_{\bob}$,
this corresponds to computing the outcome $(X_1,\allowbreak X_2)$, where, for $i \in
[2]$, $X_i = S_i$ if $v_{\bob, i}(\overline{S_i}) = 5$ and $\overline{S_i}$
otherwise.
For the rest of this section, let $f$ denote this social choice function.

\paragraph{High Level Intuition.}

We use the instance above to separate the communication complexity of EPIC
and DSIC implementations.
First, we consider the mechanism $\Mpar$ which runs
two instances of the mechanism $\M^{\star}$ from
\autoref{sec:attempt2} in parallel.
More formally, in the first round we ask Alice to report
$\{S_1, \overline{S_1}\}$ and $\{S_2, \overline{S_2}\}$, without
differentiating between sets up to complements.
Bob then picks the allocation on both sets of items in round two.
This mechanism EPIC implements $f$ with communication cost
$\BigO(m)$.
However, as we show next, $\Mpar$ fails to be DSIC.



Observe the following crucial detail of the social choice
environment: when Alice's true type is $(S_1,S_2)$,
Alice has utility $4$ when she receives $(S_1, T_2)$ for any
$T_2 \notin \{S_2, \overline{S_2} \}$, but she has utility $2$ when
she receives $(\overline{S_1}, \overline{S_2})$.
This motivates us to construct the following strategy $s_B$ of Bob in
$\Mpar$: for some fixed sets $S_1^*\subseteq M_1, S_2^* \subseteq M_2$,
if Alice reports
$\{S_1^*,\overline{S_1^*}\}$ and $\{S_2^*,\overline{S_2^*}\}$ in round one,
then Bob will give Alice $(\overline{S_1^*}, \overline{S_2^*})$.
But if Alice reports $\{S_1^*,\overline{S_1^*}\}$ and
$\{T_2,\overline{T_2}\}$ in round one, for any
$T_2 \notin \{S_2^*, \overline{S_2^*}\}$, then Bob will give Alice
$(S_1^*, T_2')$ (for $T_2' \in \{T_2, \overline{T_2} \}$
chosen arbitrarily).
When Alice's true type is $(S_1^*, S_2^*)$,
truth telling is not a best response of Alice against this strategy $s_B$.
Thus, $\Mpar$ is not DSIC.

We now argue informally that the existence of a ``crazy'' strategy like
this for Bob is not an accident, but a property which is necessary in any
communication efficient mechanism.
Intuitively, this is because for the mechanism to be efficient, Alice must
reveal a lot of information about both sets of items
(implementing \autoref{item:crazy1} from \autoref{sec:attempt1}).
Regardless of the order in which this is done, at the first point Bob
learns about Alice's type on one set of items $M_i$, he
can condition his response on the other set of items $M_{3-i}$ based on the
information from $M_i$ 
This allows him to give Alice two sets she values at $1$ when she tells the
truth, yet at least one set which she values at $4$ when she deviates
(implementing \autoref{item:crazy2}).

For a concrete example, we can also consider
$\Mseq$, which denotes the mechanism which runs $\M^{\star}$ on the
first set of items $M_1$, commits to the allocation on $M_1$, then runs
$\M^{\star}$ on the second set of items $M_2$.
Then the same argument as for $\Mpar$ shows that there is a
strategy of Bob against which truth telling is not a best response.
However, we now need to change the argument so that Bob
conditions his response on $M_2$ on Alice's actions on $M_1$,
because Bob commits to a result on $M_1$ before he acts on $M_2$.
Because Alice must reveal lots of information about her type on
\emph{both} $M_1$ and $M_2$, 
this argument should go through in any communication efficient mechanism.

\subsection{Technical Considerations and Difficulties}
\label{sec:TechnicalDifficulties}

In \autoref{sec:construction}, we argued informally that at the earliest
where Alice reveals information, it should be possible to construct a strategy of
Bob against which truth telling is not a best response for Alice.
Unfortunately, this is not literally true for every mechanism,
and our proof must circumvent this fact.
In this section, we first explain in more detail how such ``crazy''
strategies are constructed, and demonstrate that the needed ``crazy''
strategy cannot necessarily be constructed at the first node where Alice
acts.

Consider a communication efficient mechanism $\M = (G, S_A, S_B)$,
and for simplicity assume that $\M$ is perfect information\footnote{
  We prove in
  \autoref{lem:perfectInformation} 
  that this assumption is
  without loss of generality for our specific social choice 
  function $f$.
}.
This assumption allows us to not worry about situations where the mechanism
asks Alice for information, but does not reveal all of that information to
Bob.

Our goal is to construct a ``crazy strategy'' of Bob, against which
truth-telling is not a best response for Alice.
To construct this strategy, we want to find a node $h$ in the game tree of
$\M$ where
Alice communicates information which Bob can respond to in the following way:
when Alice tells the truth, Bob must be able to give Alice a bad result,
but if Alice deviates from truth telling,
Bob can give Alice a good result 
on at least one of the sets of items.
To explain this fully, we use the language of \autoref{sec:TreeTypes}.
Specifically, we use $\T_A(h),\T_B(h)$ to denote
the types of Alice and Bob for which the computation of $G$ under
truth-telling passes through $h$.
We need $h$ to satisfy the following:
\begin{enumerate}[(A)]
  \item \label{item:specificCrazy1}
    Alice acts at $h$, and there exist two of Alice's types
    $(S_1, S_2),\allowbreak (T_1, T_2)\allowbreak \in \T_A(h)$ at $h$ 
    such that $S_A((S_1, S_2))(h)\ne S_A((T_1, T_2))(h)$
    (that is, $(S_1, S_2)$ and $(T_1, T_2)$ take different actions at $h$
    under truth telling),
    and moreover, we either have $S_1 = T_1$ or $S_2 = T_2$.
    For concreteness, suppose that $S_1 = T_1$.
  \item \label{item:specificCrazy2}
    There exist types $(v_{B,1}, v_{B,2}), (v_{B,1}', v_{B,2}')
    \in \T_B(h)$ such that
    $v_{B,1}(S_1)=5, v_{B,1}(S_2)=5$, and $v_{B,1}'(S_1)=0$.
\end{enumerate}
These correspond to \autoref{item:crazy1} and
\autoref{item:crazy2} of \autoref{sec:attempt1},
instantiated for the specific social choice function $f$.

\begin{claim}
\label{claim:CrazyStratExists}
  If there exists a node $h$ at which \autoref{item:specificCrazy1}
  and \autoref{item:specificCrazy2} are both satisfied,
  then $\M$ is not DSIC.
\end{claim}
\begin{proof}
Define a ``crazy strategy'' of Bob as follows:
Bob acts according to $(v_B^1, v_B^2)$ in all nodes except those in
the subtree where Alice plays
the action chosen by $(T_1, T_2)$ at $h$, where Bob acts according to
$({v_B^1}', {v_B^2}')$. Suppose Alice's true type
is $(S_1, S_2)$. When Bob plays the above strategy and Alice tells the
truth, Alice receives $(\overline{S_1}, \overline{S_2})$,
which she values at $2$.
But if Alice deviates and plays strategy corresponding to $(T_1, T_2)$,
then she receives $S_1$ on $M_1$, and receives a utility of $4$.
Thus, truth-telling is not a best response for Alice with type $(S_1, S_2)$,
and $\M$ is not DSIC.
\end{proof}

Neither of the above conditions \autoref{item:specificCrazy1}
or \autoref{item:specificCrazy2}
on node $h$ are very strong independently.
For example, at any node $h$ which is the first time Alice takes a
nontrivial action, \autoref{item:specificCrazy1}
will be satisfied for some set $(S_1, S_2)$.
Furthermore, \autoref{item:specificCrazy2}
will be satisfied at the root node of the game tree
for \emph{every} Alice type $(S_1,S_2)$.
However, together these two requirements become somewhat subtle.
Before we proceed to the formal proof, we highlight two cases of this
subtlety, and briefly hint at how we address them.
\begin{enumerate}[(i)]
  \item \label{item:SmallSubtreeChallenge}
    Suppose the first thing the mechanism does is ask Bob ``is your type
    $(v_{B,1}^*, v_{B,2}^*)$?'' 
    (for some $(v_{B,1}^*, v_{B,2}^*)$ fixed by the mechanism).
    If the answer is yes, then all types of Alice have a dominant strategy
    in the corresponding subtree.
    Moreover, if the first question is to just ask Bob ``is your full type
    on $M_1$ equal to $v_{B,1}^*$?'' (for some fixed $v_{B,1}^*$,
    regardless of his type on $M_2$) then it is possible
    that Alice always has 
    a dominant strategy in that subtree\footnote{
      Observe that Alice already knows what will happen on $M_1$.
      Thus, in this subtree the mechanism can thus run
      the DSIC mechanism $\M^{\star}$ described in \autoref{sec:attempt1}
      on $M_2$.
      Then, as a final step the mechanism can ask Alice her type on $M_1$.
      Intuitively, Alice already knows what will happen on $M_1$
      (and can always grantee her best attainable outcome on $M_1$ at the end),
      so she might as well try to get her full value on $M_2$.
    }.
    This shows that we cannot hope to construct the needed ``crazy
    strategy'' of Bob in every subtree of the game.
  \item \label{item:SmallQuestionChallenge}
    Suppose the first question is to ask Bob ``what is your value on
    sets $\{T_1^*, \overline{T_1^*}\} \subseteq M_1$
    and sets $\{T_2^*, \overline{T_2^*}\} \subseteq M_2$ (for some
    set $T_1^*, T_2^*$ fixed by the mechanism).
    At the (four nodes of the) next layer of the tree, ask Alice
    ``Do you have $S_1 \in \{T_1^*, \overline{T_1^*} \}$
    AND $S_2 \in \{T_2^*, \overline{T_2^*} \}$?''
    It turns out that truth-telling is a dominant action at every node in
    the first layer where Alice acts\footnote{
      Formally, truth-telling is a dominant action at node $h$ if
      $S_A(t_A)$ gets utility at least as high as all strategies $s_A'$
      such that $s_A'(h)\ne S_A(t_A)(h)$.

      Clearly Alice has a dominant strategy if indeed she should
      answer ``yes'' in this layer.
      If not, either one or both of her sets are not in the
      specified pair. If both are not, she gets zero utility from lying.
      If one of her sets is in the specified pair,
      the outcome on the matching set of items is already fixed, so Alice
      might as well ``continue'' (answering ``no'') and hope for more utility
      on the other set of items, knowing she can always grantee her utility on
      the matching set of items.
    }.
    This shows that we cannot hope to construct the needed ``crazy
    strategy'' of Bob at every layer of the game tree.
\end{enumerate}

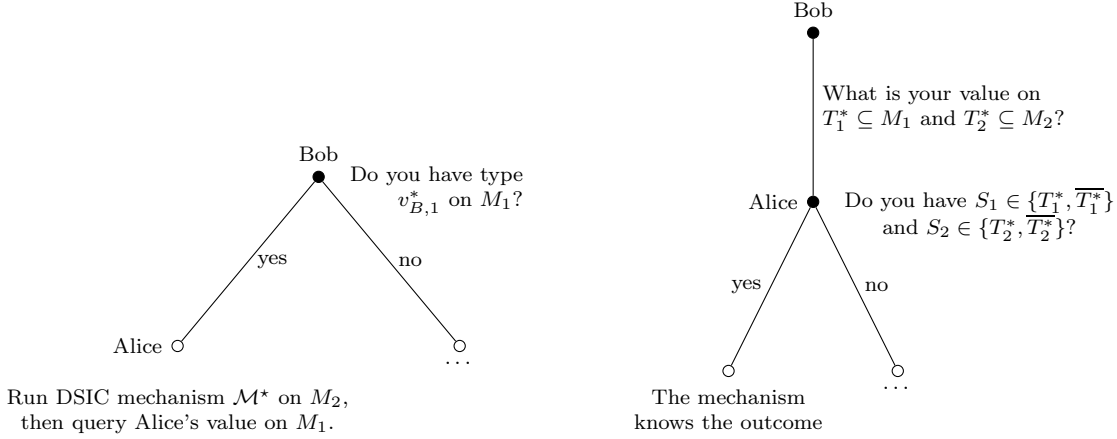
\begin{figure}
  \tikzset{
    solid node/.style={circle,draw,inner sep=1.5,fill=black},
    hollow node/.style={circle,draw,inner sep=1.5}
  }
  \begin{tikzpicture}[scale=1.5,font=\footnotesize]
    \tikzstyle{level 1}=[level distance=15mm,sibling distance=25mm]
    \tikzstyle{level 2}=[level distance=15mm,sibling distance=15mm]
    \tikzstyle{level 3}=[level distance=15mm,sibling distance=7mm]
  \node(0)[solid node,label=above:{Bob},
      label=right:{\begin{tabular}{r}
          \\
          Do you have type \\
          $v_{B,1}^*$ on $M_1$?
        \end{tabular} } 
      ]{}
    child{node(1)[align=left, hollow node, label=left:{Alice},
      label=below:{\begin{tabular}{c}
          \\
          Run DSIC mechanism $\M^{\star}$ on $M_2$, \\
          then query Alice's value on $M_1$.
        \end{tabular} } 
      ]{}
      edge from parent node[align=left, right,xshift=0]{yes}
    }
    child{node(2)[align=left, hollow node, label=below:{\ldots}]{}
      edge from parent node[align=left, right,xshift=0]{no};
    } ;
  \end{tikzpicture}
  \qquad
  \begin{tikzpicture}[scale=1.5,font=\footnotesize]
    \tikzstyle{level 1}=[level distance=15mm,sibling distance=35mm]
    \tikzstyle{level 2}=[level distance=15mm,sibling distance=15mm]
    \tikzstyle{level 3}=[level distance=15mm,sibling distance=7mm]
  \node(0)[solid node,label=above:{Bob}]{}
    child{node(1)[align=left, solid node, label=left:{Alice},
      label=right:{\begin{tabular}{c}
          \\
        Do you have $S_1 \in \{T_1^*, \overline{T_1^*} \}$ \\
        and $S_2 \in \{T_2^*, \overline{T_2^*} \}$? 
        \end{tabular} } ]{}
    child{node[hollow node,label=below:{\begin{tabular}{c}
          The mechanism \\
          knows the outcome
        \end{tabular} } ]{}
        edge from parent node[left]{yes}
    }
    child{node[hollow node,label=below:{\ldots}]{} 
        edge from parent node[right]{no}
    }
    edge from parent node[align=left, right,xshift=0]{What is your value on
        \\ $T_1^*\subseteq M_1$ and $T_2^*\subseteq M_2$? \\ }
    } ;
  \end{tikzpicture}
  \caption{Examples of the technical difficulties our proof needs to
  handle. These figures illustrate the first few layers of the game trees,
  while the remainder of the game is unspecified.
  Regardless of how the rest of the game computes $f$, 
  these examples illustrate that the required strategy of Bob cannot always 
  be constructed based on the first node where Alice acts (i.e. the first
  node where Alice acts may not satisfy \autoref{item:specificCrazy1} and
  \autoref{item:specificCrazy2} for any $(S_1,S_2), (T_1,T_2)$).  }
  \label{fig:difficulties} 
\end{figure}

Intuitively, we address the first issue by noting that, because Bob has more types
than there are nodes in the game tree, we can safely ignore any node
in which Bob has few types.
We fix the second issue by changing the proof outline overall.
Instead of taking an efficient mechanism $\M$ and finding a node $h$ satisfying
\autoref{item:specificCrazy1} and \autoref{item:specificCrazy2}
(thus showing that $\M$ is not DSIC), we use a proof by contradiction.
Intuitively, we consider an efficient mechanism in which no such
``crazy strategy'' of Bob be constructed, and
show that the questions such a mechanism can ask to Alice are so
restrictive that the mechanism cannot possibly handle all types Alice might
have.

%% file: DsicVsEpicFormalities.tex
\subsection{Separation of DSIC and EPIC without Transfers}
\label{app:DsicVsEpic}


We now prove that, without transfers, social choice function 
$f$ from \autoref{sec:construction} requires an 
exponential amount of communication to implement in dominant strategies.

\begin{theorem}
\label{thrm:EpicVsDsicNoTransfers}
  Any DSIC implementation of $f$ without transfers has communication cost
  $\widetilde\Omega(2^{m})$.
\end{theorem}

\begin{proof}

Fix a DSIC mechanism $\M$ and let $\comm$ be the communication of
$\M$. At the cost of blowing up the communication by a factor of two, we can
assume by 
\autoref{lem:perfectInformation} 
that $\M$ is perfect information.
Let $\M = (G, S_A, S_B)$, i.e. $G$ is the perfect information extensive form
game used by $\M$, and $S_A, S_B$ are the dominant type-strategy 
profile implementing $f$.
By \autoref{sec:TreeTypes},
each node $h$ of the game tree $G$ corresponds to a set of types
$\T_A(h)$ of Alice and $\T_B(h)$ of Bob, and each action taken at $h$
corresponds to partitioning the types of the player $\mathcal{P}(h) \in
\{A, B\}$ who acts at this node.
This partition is given by $S_i(\cdot)$, specifically,
for each node $h'$ immediately after $h$ in $G$,
a type $t\in \T_i(h)$ remains in $\T_i(h')$ if and only if $S_i(t)$
plays the action at $h$ which corresponds to $h'$.
We say that type $t$ takes action $a$ at $h$ if $S_i(t)$ plays $a$ at $h$.
Observe that for each $i$ and nodes $h, h'$ where $h$ is an ancestor
of $h'$, we have $\T_i(h') \subseteq \T_i(h)$.

Define $K := \binom{m}{m/2} \big/ 2$ and observe that $K =
\widetilde\Theta(2^m)$. Observe that, for $i \in [2]$, we can partition all
subsets of $M_i$ of size $m/2$ into (unordered) pairs of the form $(T,
\overline{T})$, and there will be exactly $K$ such pairs. Call these pairs
$P_{i,1}, \cdots, P_{i,K}$ in some canonical order. For the rest of this
section, we equivalently view a type $(v_{\bob, 1}, v_{\bob, 2})$ of Bob as
a pair of bit-strings $B = (B_1, B_2) \in \{0,1\}^K \times \{0,1\}^K$,
where $B_{i,k}$ for $k \in [K]$ specifies which set in $P_{i,k}$ Bob values
at $5$. Similarly, we can view Alice's type $(S_1, S_2)$ as a tuple $(k_1,
k_2, b_1, b_2)$ where $k_1, k_2 \in [K]$ are indices and $b_1, b_2$ are
bits, and, for $i \in [2]$, $S_i$ is the $b_i^{\text{th}}$ element in
$P_{i,k_i}$. As $b_1, b_2$ are irrelevant to the outcome of the mechanism,
using
\autoref{lem:noIrrelevantQuestions}\footnote{
  More formally, consider the partition of Alice's types given by
  $\{  \{(k_1,k_2,b_1,b_2)\}_{(b_1,b_2)\in\{0,1\}\times\{0,1\}}   
    \}_{(k_1,k_2)\in[K]\times[K]} $.
  For all fixed types of Bob, $f$ is constant on the above partition,
  and thus by \autoref{lem:noIrrelevantQuestions},
  we can assume that for all $h$, if one element of a set 
  $\{(k_1,k_2,b_1,b_2)\}_{(b_1,b_2)\in\{0,1\}\times\{0,1\}}$ is
  in $\T_A(h)$, then all elements of that set are in $\T_A(h)$.
},
we can assume without loss of
generality that for all $k_1,k_2 \in [K]$, and nodes $h$, we
either have that all Alice's types of the form $(k_1, k_2, \cdot, \cdot)
\in \T_{\alice}(h)$ or all types of the form $(k_1, k_2, \cdot, \cdot)
\notin \T_{\alice}(h)$. Thus, when talking about the sets $\T_{\alice}(h)$
for nodes $h$, we can view Alice's type as simply a pair of indices $(k_1,
k_2)$. We adopt this convention for the rest of this proof, and we consider
$\T_{\alice}(h)\subseteq [K]\times[K]$. Correspondingly,
we consider $S_A(\cdot)$ to be
a map from $[K]\times[K]$ to strategies in $\M$,
and refer to the actions taken by pairs $(k_1,k_2)\in[K]\times[K]$.

The social choice function $f$ is determined by Alice's
index on both sets of items, as well as Bob's value on those two indices.
Thus, for each leaf $\ell$, we have $\T_A(\ell) = \{(k_1,k_2)\}$ a
singleton\footnote{
    Bob's type, on the other hand, need only be determined on
    indices $k_1, k_2$.
    That is, for each leaf $\ell$, if we have $\T_A'(\ell) = \{(k_1,k_2)\}$,
    then for each $(B_1,B_2),(B_1',B_2') \in \T_B(\ell)$,
    we have $B_{1,k_1}=B_{1,k_1}'$ and 
    $B_{2,k_2}=B_{2,k_2}'$.
}, which will be a key observation in our proof. 


We now begin to build the language and tools needed to address the
considerations highlighted in \autoref{sec:TechnicalDifficulties}.

\begin{definition}[Shattered pairs]
\label{def:shattered}
Let $h$ be a node and $(k_1, k_2) \in [K] \times [K]$. We say that $(k_1,
k_2)$ is \emph{shattered} at $h$ if 
Bob's types $B = (B_1, B_2) \in \T_{\bob}(h)$, when restricted to
coordinates $k_1, k_2$, take on all four possible values. In
other words,
\[
\card*{\{ (B_{1,k_1}, B_{2,k_2}) \mid (B_1, B_2) \in \T_{\bob}(h) \}} = 4.
\]
We use $\T^{\shat}(h)\subseteq [K]\times[K]$ to denote the set of all pairs
$(k_1, k_2)$ that are shattered at $h$.
\end{definition}


For convenience, we define the ``neighbors'' of a pair $(k_1,k_2)$ to be
all those pairs with (at least) one index in common with $(k_1, k_2)$.
Note that $(k_1,k_2)\in\nbr(k_1,k_2)$.
\begin{definition}[Neighbors]
\label{def:nbr}
Let $(k_1, k_2) \in [K] \times [K]$. 
A \emph{neighbor} of $(k_1, k_2)$ is any pair of the form
form $(k_1, k')$ or $(k', k_2)$ for some $k' \in [K]$.
We use $\nbr(k_1, k_2)$ to denote the set of all neighbors of $(k_1, k_2)$.
\end{definition}

Our first two lemmas show that a pair being shattered at a node $h$
severely restricts which questions the mechanism can ask at $h$.
The first lemma corresponds to \autoref{claim:CrazyStratExists},
recast in the language of this proof.
More specifically, \autoref{item:specificCrazy2} from 
\autoref{sec:TechnicalDifficulties} corresponds to a pair $(k_1,k_2)$
being shattered at $h$, and \autoref{item:specificCrazy1} from 
\autoref{sec:TechnicalDifficulties} corresponds to $(k_1,k_2)$ taking a
different action than one of its neighbors.
These two items cannot simultaniously occur in a DSIC mechansism.
\begin{lemma}
\label{lem:ShatSameAction}
Consider any node $h$ with $\mathcal{P}(h)=\alice$
and shattered pair $(k_1,k_2)\in\T^{\shat}(h) \cap \T_A(h)$.
Then every pair in $\nbr(k_1,k_2) \cap \T_A(h)$ must take the same action
as $(k_1,k_2)$ at $h$.
\end{lemma}
\begin{proof}
Fix a $(k_1, k_2) \in \T^{\shat}(h) \cap \T_A(h)$.
Suppose for contradiction that there
exists a neighbor of $(k_1, k_2)$ which is in $\T_A(h)$, yet takes a
different action from $(k_1, k_2)$ at $h$. Without loss of generality,
assume this neighbor is of the form $(k_1, k_2')$.
We derive a contradiction by
constructing a ``crazy'' strategy for Bob, exactly as in
\autoref{claim:CrazyStratExists}, that violates the DSIC property.

Pick some $(B_1,B_2), (B_1',B_2') \in \T_B(h)$ with
$B_{1,k_1} = B_{2,k_2} = 1$ and $B_{1,k_1}'=0$
(these exist by \autoref{def:shattered}).
We define a strategy $s_B$ of Bob such that in subtree where Alice 
plays the action taken by $(k_1, k_2')$ at $h$,
$s_B$ plays the action played by $S_B( (B_1',B_2') )$.
In every other node of the game tree, $s_B$ plays the same 
action played by $S_B( (B_1, B_2) )$.
This completely specifies $s_B$.

Suppose Alice has type $(k_1, k_2, 1, 1)$ (that is, for $i\in[2]$,
her desired sets are in $P_{i,k_i}$, and her most preferred set
is the one from $P_{i,k_i}$ which $(B_1,B_2)$ values at $5$).
When Alice plays $S_A( (k_1, k_2) )$, she is allocated her less preferred
set on both $M_1$ and $M_2$, and thus gets utility $2$.
But if Alice deviates and plays $S_A( (k_1, k_2') )$,
on $M_1$ she receives her most preferred set, which she values at $4$.
Thus, truth telling is not a best response for Alice with type 
$(k_1, k_2)$, and thus $\M$ is not DSIC.
\end{proof}

We just showed that if some pair is shattered and lies in some $\T_A(h)$,
then all of its neighbors who are also in $\T_A(h)$ must take the same
action.
Next, we need something stronger, namely that all shattered pairs at $h$
take the same action.
Along the way we prove that additionally that if a pair is shattered and
lies in $\T_A(h)$, then all of its neighbors must lie in $\T_A(h)$.

\begin{lemma}
\label{lem:ShatNbrSameAction}
Let $h$ be any node, and consider the set
\[
\T_{\alice}^{\shatnbr}(h) := 
  \bigcup_{(k_1, k_2) \in \T^{\shat}(h) \cap \T_A(h)} \nbr(k_1, k_2) .
\]
Then we have that $\T_{\alice}^{\shatnbr}(h) \subseteq \T_{\alice}(h)$.
Moreover, if $\mathcal{P}(h) = \alice$, then all types in 
$\T_A^{\shatnbr}(h)$ must take the same action at $h$.
\end{lemma}
\begin{proof}


First, we prove that $\T_A^{\shatnbr}(h)\subseteq \T_{\alice}(h)$.
Suppose for contradiction that this is not the case.
This means that 
there exists $(k_1,k_2)\in \T^{\shat}(h) \cap \T_A(h)$ and 
$(k_1',k_2')\in \nbr(k_1,k_2)$, but $(k_1',k_2')\notin\T_A(h)$.

Consider the node $h'$ which is the latest Alice node along the
path from the root to $h$ at which $(k_1',k_2')\in\T_A(h')$.
By definition, $(k_1,k_2)$ and $(k_1',k_2')$ take different
actions at $h'$.
Observe that, because $\T_B(h)\subseteq\T_B(h')$,
we also have $\T^{\shat}(h)\subseteq\T^{\shat}(h')$, and thus
$(k_1,k_2)$ is shattered at $h'$.
But then, by \autoref{lem:ShatSameAction}, $(k_1,k_2)$
and $(k_1',k_2')$ must take the same action at $h'$,
a contradiction.

Conclude using $\T_A^{\shatnbr}(h)\subseteq \T_{\alice}(h)$ and \autoref{lem:ShatSameAction} that,
if $(k_1,k_2) \in \T^{\shat}(h)\cap \T_A(h)$,
then each pair in $\nbr(k_1,k_2)$ takes the same action at $h$.
Thus, to prove that all types in $\T_{\alice}^{\shatnbr}(h)$ take
the same action, it suffices to show that if
$(k_1,k_2), (k_1',k_2') \in \T^{\shat}(h)\cap \T_A(h)$
with $k_1\ne k_1'$ and $k_2\ne k_2'$, then $(k_1,k_2)$ and $(k_1',k_2') $
must still take the same action at $h$.

To prove this, we make use of the fact that 
$\T_{\alice}^{\shatnbr} \subseteq \T_A(h)$.
In particular, means that $(k_1, k'_2) \in \nbr(k_1, k_2) \subseteq \T_A(h)$.
By \autoref{lem:ShatSameAction}, 
$(k_1, k'_2)$ must take the same action as $(k_1, k_2)$.
By the exact same logic, $(k_1, k'_2)$ must take the same 
action as $(k'_1, k'_2)$.
Thus, $(k_1,k_2)$ and $(k_1',k_2')$ take the same action at $h$,
and so does every pair in $\T_{\alice}^{\shatnbr}(h)$.

\end{proof}

While $\T^{\shat}(h)\subseteq [K]\times[K]$ is defined entirely in terms of
Bob's types, $\T_{\alice}^{\shatnbr}(h)\subseteq \T_{\alice}(h)$
tells us information about Alice's types as well.
This provides us with a convenient way to describe the rest of the proof,
in terms of the following observation:

\begin{observation}
\label{obs:ShatNbrEmptyAtLeaves}
For all leaves $\ell$ of the game tree, we have
$\T_{\alice}^{\shatnbr}(\ell) = \emptyset$
(that is, ${ \T^{\shat}(\ell)\cap\T_{\alice}(\ell) = \emptyset }$).
\end{observation}
\begin{proof}
    Recall that social choice function $f$ is determined by Alice's
    index on both sets of items, as well as Bob's value on those two indices.
    In particular, at every leaf node, the mechanism must completely
    know Alice's pair in order to correctly compute $f$.
    Thus, for all leaves $\ell$ of the game tree,
    $\T_{\alice}(\ell)$ is a singleton.
    Using \autoref{lem:ShatNbrSameAction} and the fact that 
    $|\nbr(k_1,k_2)|>1$,
    this is possible only if all leaves $\ell$ satisfy
    $\T_{\alice}^{\shatnbr}(\ell) = \emptyset$.
\end{proof}
Our task in the remainder of the proof is to show that, if the communication
cost $\comm$ of $\M$ is sufficiently small, then there must exist a leaf $\ell$
with $\T_{\alice}^{\shatnbr}(\ell) \ne \emptyset$.

Note that it is possible for a mechanism with exponential communication 
to satisfy $\T_{\alice}^{\shatnbr}(\ell) = \emptyset$ at every leaf. 
Indeed, in the direct revelation mechanism where Bob reveals his entire
type, $\T_{\bob}(\ell)$ is singleton for every leaf,
and thus $\T^{\shat}(\ell) = \emptyset$.
However, our next two lemmas shows that in low-communication mechanisms,
very few types of Bob can ever end up at leafs $\ell$ in which 
$|\T^{\shat}(\ell)|$ is small.
This serves to address \autoref{item:SmallSubtreeChallenge} from
\autoref{sec:TechnicalDifficulties}.
Our next lemma is a standard communication complexity argument,
and intuitively states that ``typical'' Bob types always end up in leaves
with a large number of Bob types.


\begin{lemma} 
  \label{lem:AtypicalBobTypes}
  Let $\TypicalBob\subseteq \T_{\bob}$ denote the set of all 
  Bob types $B \in \T_{\bob}$ such that for all leaves $\ell$ with 
  $B \in \T_{\bob}(\ell)$, we have 
  $\card*{\T_{\bob}(\ell)} \geq \card*{\T_{\bob}} \cdot 4^{-2 \comm}$.
  We have
  \[ \card*{\TypicalBob} \geq \card*{\T_{\bob}} \cdot (1 - 4^{-\comm})
  \]
\end{lemma}
\begin{proof}
We show that $\card*{\T_{\bob} \setminus \TypicalBob} \leq 4^{-\comm} \cdot
\card*{\T_{\bob}}$. Indeed, for all $B \in \T_{\bob} \setminus \TypicalBob$,
there exists a leaf $\ell$ such that $B \in \T_{\bob}(\ell)$
and $\card*{\T_{\bob}(\ell)} < \card*{\T_{\bob}} \cdot 4^{-2 \comm}$.
There are at most $2^{\comm}$ leaves in $G$.
This gives:
\begin{align*}
\card*{\T_{\bob} \setminus \TypicalBob} 
\leq \sum_{\substack{\text{leaf $\ell$ such that}
  \\\card*{\T_{\bob}(\ell)} < \card*{\T_{\bob}} \cdot 4^{-2 \comm}}
  }
  \card*{\T_{\bob}(\ell)}
\leq \sum_{\text{leaf } \ell} \card*{\T_{\bob}} \cdot 4^{-2 \comm} 
\leq \card*{\T_{\bob}} \cdot 4^{-\comm}.
\end{align*}
\end{proof}

Next, we show that in all ``typical nodes'' (that is, nodes containing
even one type from $\TypicalBob$), \emph{most} of the pairs in
$[K]\times[K]$ are shattered. 
This follows from a combinatorial argument -- if a lot of pairs are
\emph{not} shattered at $h$, then there cannot possibly be enough types in
$\T_{\bob}(h)$ for $h$ to be ``typical''.

\begin{lemma}
  \label{lem:UnshatteredCounting}
  For any node $h$ such that
  $\T_{\bob}(h)\cap \TypicalBob \ne \emptyset$,
  we have $\card*{\T^{\shat}(h)} \geq K^2 - 10 K\comm$.
\end{lemma}
\begin{proof}
  Fix a node $h$ with $\T_{\bob}(h)\cap \TypicalBob \ne \emptyset$.
  Some descendent of $h$ is a leaf node $\ell$ where  
  $\T_{\bob}(\ell)\cap \TypicalBob \ne \emptyset$.
  By the definition of $\TypicalBob$, this means
  $\card*{\T_{\bob}(\ell)} \ge |\T_{\bob}|\cdot 4^{-2\comm}$. 
  Thus $|\T_{\bob}(h)| \ge |\T_{\bob}|\cdot 4^{-2\comm}$ as well.
  This condition will suffice to bound $\card*{\T^{\shat}(h)}$.

  We consider the set $\overline{\T^{\shat}(h)}
  =  [K] \times [K] \setminus \T^{\shat}(h)$ of \emph{unshattered} pairs,
  and proceed by showing that
  $\card{ \overline{\T^{\shat}(h)} }\leq 10 K\comm$.
  To this end, observe that 
  each unshattered pair $(k_1, k_2)$ can be uniquely written
  as $(k, k+d)$ for some values of $k, d \in [K]$ (where we take indexes mod $K$). To show that $\card{ \overline{\T^{\shat}(h)} }\leq 10 K\comm$, we actually show that for all $d \in [K]$, the number of pairs of the form $(k, k+d) \in \overline{\T^{\shat}(h)}$  is at most $10 \comm$. Summing over all $d \in [K]$ then proves the lemma.

  Fix a $d$ and suppose for contradiction that 
  there were more than $R = 10\comm$ pairs in
  $\overline{\T^{\shat}(h)}$ of the form $(k, k+d)$.
  For any fixed $d$, all of the $4^K$ types of Bob in $\T_B$
  can be uniquely described by specifying Bob's value
  on $(k, k+d)$ for each $k\in[K]$,
  that is, by specifying for each $k\in[K]$ 
  one of the four possible values
  of $(B_{1, k}, B_{2, k+d})\in \{0,1\}\times\{0,1\}$.

  If $(k, k+d)$ is shattered at $h$, then the types of Bob in
  $\T_B(h)$ can take on all $4$ possible values on indexes $(k, k+d)$.
  However, if $(k, k+d)$ is unshattered at $h$, then there is at least one
  of the $4$ options for Bob's type on indexes $(k,k+d)$
  which never occurs in $\T_B(h)$.
  Thus, the types of Bob in
  $\T_B(h)$ can take on at most $3$ possible values on indices $(k,k+d)$.
  Thus, the number of types in $\T_B(h)$ satisfies
  \[ \card*{\T_{\bob}(h)} 
    \le 3^R \cdot 4^{K - R} = 4^{K - (1 - \log_4 3) R} < 4^{K - (1/5)R}.
  \]
  Plugging $R= 10 \comm$, we have
  $\card*{\T_{\bob}(h)} < 4^{K - 2 \comm} = |\T_B|\cdot 4^{-2\comm}$,
  which contradicts what we know about $|\T_B(h)|$.

\end{proof}

Even in communication efficient mechanisms, there can be leaf nodes with
$\T^{\shatnbr}_A(\ell) = \emptyset$.
The right hand side of \autoref{fig:difficulties},
illustrating \autoref{item:SmallQuestionChallenge}
in \autoref{sec:TechnicalDifficulties}, gives an example.
Looking into this example deeper, we see the reason: at the node $h$ 
where Alice acts, if we take the action \emph{not taken by the pairs in
$\T^{\shat}(h)$}, then we arrive at a leaf node with 
$\T^{\shat}(\ell) \cap \T_{\alice}(\ell) = \emptyset$.
Thus, intuitively, our approach for the remainder of this proof
is to follow the actions taken by $\T^{\shat}(h)$
in order to arrive at a node with
$\T_{\alice}^{\shatnbr}(\ell) \ne \emptyset$.

We now begin to wrap up our proof.
Assume for contradiction that $\comm < K/10$. By \autoref{lem:AtypicalBobTypes}, 
we get that $\comm < K/10$ implies that $\TypicalBob \neq \emptyset$. 
Fix an arbitrary $B^{\star} \in \TypicalBob$.
From \autoref{lem:UnshatteredCounting}, for any node $h$ such that 
$B^{\star} \in \T_{\bob}(h)$, we have $\T^{\shat}(h) \neq \emptyset$. 

Define a collection of nodes $H^{\star}$ in $G$ as follows:
\[ H^{\star} = \{ h \mid B^{\star} \in \T_{\bob}(h),
  \ \T^{\shat}(h)\subseteq \T_A(h) \}
\]
Observe that the root $h_0$ of $G$ is in $H^{\star}$
(because $\T_{\alice}(h_0) = \T_{\alice}$ and  $\T_{\bob}(h_0) = \T_{\bob}$),
and thus $H^{\star}\ne \emptyset$.
Now, define $h^{\star}$ to any node in $H^{\star}$ for which no descendent of 
$h^{\star}$ is in $H^{\star}$\footnote{
  One can use \autoref{lem:ShatNbrSameAction} to show that $h^{\star}$ is
  unique and $H^*$ forms a path from the root to a leaf.
  However, this is not needed for our argument to go through.
}.


First, we claim that $h^{\star}$ cannot be a node where Bob acts.
Otherwise, consider the child $h'$ of $h^{\star}$ corresponding
to the action taken by $B^{\star}$ at $h^{\star}$.
At $h'$, Alice's type set remains the same, while sets $\T_B(h')$ and thus 
$\T^{\shat}(h')$ have only decreased from $h^{\star}$.
Thus,
\[
\T^{\shat}(h') \subseteq \T^{\shat}(h^{\star}) 
  \subseteq \T_{\alice}(h^{\star}) = \T_{\alice}(h'),
\]
and $h'\in H^{\star}$. This contradictions the choice of $h^{\star}$. 

Next, we claim $h^{\star}$ cannot be a node where Alice acts.
Suppose otherwise. Because
$\T^{\shat}(h^{\star})\subseteq\T_A(h^{\star})$,
we have $\T^{\shat}(h^{\star})\subseteq\T_{\alice}^{\shatnbr}(h^{\star})$.
Because $\T^{\shat}(h^{\star})\ne\emptyset$,
we have $\T_{\alice}^{\shatnbr}(h^{\star})\ne\emptyset$.
By \autoref{lem:ShatNbrSameAction}, there is thus a single child 
$h'$ of $h^{\star}$ such that every pair in $\T_{\alice}^{\shatnbr}(h^{\star})$
takes the action leading to $h'$. 
As Bob's type set is unchanged at $h'$, we thus get
\[
\T^{\shat}(h') = \T^{\shat}(h^{\star})
\subseteq \T_{\alice}^{\shatnbr}(h^{\star}) \subseteq \T_{\alice}(h'),
\]
and $h'\in H^{\star}$. This contradictions the choice of $h^{\star}$. 

This means that $h^{\star}$ must be a leaf node.
But then we have $\emptyset \neq \T^{\shat}(h^{\star}) \subseteq \T_{\alice}(h^{\star})$,
and thus $\T_{\alice}^{\shatnbr}(h^*) \ne \emptyset$.
This contradicts \autoref{obs:ShatNbrEmptyAtLeaves}.

Thus, for any mechanism $M$ which DSIC implements $f$ without transfers,
$H_M \ge (1/10)K = \widetilde\Omega(2^m)$.

\end{proof}

We remark that this theorem is tight up to constants, as the direct
revelation mechanism asking Bob to reveal his entire type has communication cost $\mathcal{O}(K)$.

%% file: DsicVsEpicWithTransfers.tex
\subsection{Extension to the Case with Transfers}
\label{sec:DsicVsEpicTransfers}

We have shown that social choice function $f$, which is incentive compatible
without transfers, cannot be efficiently DSIC implemented without transfers.
However, this does not yet rule out the existence of certain
transfer functions which can efficiently DSIC implement $f$\footnote{
  Note that in principle it is possible for certain transfer functions to
  make a mechanism DSIC, but for others to render a mechanism EPIC but not
  DSIC. For example, suppose Alice has two types $L, R$ and Bob has two types
  $a, b$, and we have $f(L, a) = 1, f(L, b)=2, f(R,a)=3, f(R,b)=4$.
  Let Alice with type $L$ value $1,2,3,4$ at $10, 7, 8, 1$ respectively,
  and Alice with type $R$ value $1,2,3,4$ at $0, 0, 10, 10$ respectively.
  Bob's valuations are irrelevant.
  Consider the perfect information mechanism sequentially asking Alice for
  her type, then Bob for his.
  If no transfers are included (that is, all transfers are $0$) then this
  mechanism is EPIC but not DSIC.
  If the transfers to Alice when the outcome is $1,2,3,4$ are $0,2,0,0$
  respectively (that is, we pay Alice $2$ when outcome $2$ is selected),
  then this mechanism is DSIC.
  We do not know if $f$ has an efficient DSIC implementation with
  transfers, but we suspect it does not.
}.

\paragraph{Modified construction.}
We now describe how to modify our construction
to prove our separation even in the regime with transfers.
For each $i\in[2]$ and each pair of subsets of items of the form
$\{T, \overline{T} \} \subseteq M_i$ where $|T|=m/2$,
we add a type of Alice which values every set at $0$.
The outcome when Alice has such a type on $M_i$ is identical to if Alice
had positive utility for $T$ and $\overline{T}$ originally.
That is, Bob receives whichever set among $\{T, \overline{T} \}$
he values at $5$, and Alice receives the complement.

More formally, the set of outcomes and Bob's types (and Bob's utility for
each outcome) remains unchanged, but Alice's type set changes.
For $i\in[2]$, we let $\mathcal{Q}_i$ denote the collection of all subsets
of $M_i$ of size $m/2$, and let
$\mathcal{R}_i$ denote the collection of all
unordered pairs of subsets of $M_i$ of the form $\{T, \overline{T}\}$
with $|T|=m/2$.
Alice's new set of types are then $\T_A' =
(\mathcal{Q}_1\cup\mathcal{R}_1) \times (\mathcal{Q}_2\cup\mathcal{R}_2)$.
We use $\mathcal{R}_i$ to represent the cases where Alice gets $0$ value
from sets of items in $M_i$.
Specifically, Alice's utility function
$u_{\alice}' : \T_{\alice}' \times Y \to \mathbb{R}$
is now defined by $u_{\alice}'((S_1, S_2), (X_1, X_2))
= u_{\alice, 1}'(S_1, X_1) + u_{\alice, 2}'(S_2, X_2)$,
where, for $i \in [2]$, we have:
\[
 u_{\alice, i}'(S_i, X_i) = \begin{cases}
 4, &\text{~if~} S_i \in \mathcal{Q}_i \text{~and~} S_i = X_i \\ 
 1, &\text{~if~} S_i \in \mathcal{Q}_i \text{~and~} S_i = \overline{X_i} \\
 0, &\text{~otherwise}.
 \end{cases}
\]
In particular, $u_{A,i}'(S_i,X_i) = 0$ whenever $S_i \in \mathcal{R}_i$.

We define the social choice function $f'$ as follows:
Let $f'\big( (S_1, S_2), \allowbreak (v_{\bob, 1}, v_{\bob, 2}) \big) = (X_1, X_2)$,
where for $i\in[2]$, if $S_i \in \mathcal{R}_i$, then $X_i \in S_i$ is such
that $v_{\bob, i}(\overline{X_i}) = 5$, and if $S_i \in \mathcal{Q}_i$,
then $X_i = S_i$ if $v_{\bob, i}(\overline{S_i}) = 5$ and
$X_i = \overline{S_i}$ if $v_{\bob, i}({S_i}) = 5$.
Observe that $f'$ is still welfare maximizing (although we use a very
specific tie-breaking rule for those Alice types with $S_i \in
\mathcal{R}_i$).

Intuitively,
the addition of Alice types which are irrelevant to all outcomes allows us
to say that the mechanism cannot award transfer to Alice in a nontrivial way.
This allows us to reduce to the case without transfers.
We make this formal below.


\begin{restatable}{theorem}{restateDsicVsEpic}\label{thm:DsicVsEpic}
  There exists an EPIC implementation of $f'$ with communication cost
  $\BigO(m)$.
  However, any DSIC implementation of $f'$ with transfers has communication cost
  $\widetilde\Omega(2^m)$. That is,
  \[ CC^{EPIC}(f') = \BigO(m) \qquad\qquad 
    CC^{DSIC}(f') = \widetilde\Omega(2^m).
  \]
\end{restatable}
\begin{proof}

Consider the mechanism which asks Alice to reveal her entire type,
tells that type to Bob, and asks Bob to choose an outcome.
All transfers are $0$.
This is EPIC, for the exact same reason that $\Mpar$
in \autoref{sec:construction} is EPIC.
Moreover, the communication cost is $\BigO(m)$, as desired.

Now, consider a mechanism $\M_0$ which DSIC implements $f'$ with transfers.

In principle, this mechanism may provide nonzero transfers to Bob
(specifically, if Alice acts at the root node, this action
may change the transfer to Bob arbitrarily).
However, observe that if we replace every transfer to Bob with $0$, the
result is still DSIC. This is because incentives have changed only for Bob,
but Bob can now guarantee himself utility $10$ when he follows
$S_{\bob}(t_{\bob})$ (and this is the highest utility he can achieve).
Let $\M_1$ denote the mechanism that sets every transfer to Bob in 
$\M_0$ to $0$.

Now that Bob has constant utility in every outcome selected by $\M_1$,
let $\M_2$ denote the result of applying
\autoref{lem:perfectInformation}
to $\M_1$ to get a
perfect information DSIC mechanism.
This leaves the transfers unchanged, and effects the communication cost by
only a constant factor.


We describe and partition Alice's types $\T_A'$ in a similar way to how
we partitioned them in the proof of \autoref{thrm:EpicVsDsicNoTransfers}.
Let $K := \binom{m}{m/2} \big/ 2$ and, for $i\in[2]$,
recall that $\mathcal{R}_i$ denotes a partition of subsets of 
$M_i$ of size $m/2$ into pairs $\{T, \overline{T}\}$.
Index the pairs in this partition by $k\in[K]$.
Then the type of Alice can be described by tuple $(k_1, k_2, v_1, v_2)$ 
where $k_1, k_2 \in [K]$ and 
$v_1, v_2 \in \{0,1,2\}$.
Specifically, for $i\in[2]$, the index $k_i$ specifies which set in
$\mathcal{R}_i$ Alice's type corresponds to,
and $v_i\in\{0,1,2\}$ specifies 
whether Alice's most preferred set is
$T$ (when $v_i=0$), $\overline{T}$ (when $v_i=1$),
or neither (i.e. if $v_i = 2$, then Alice receives $0$ utility
from all subsets of $M_i$).

Observe that for every type of Bob,
$f'$ is independent of the values $v_1,v_2$,
and depends only on the indices $(k_1,k_2)$.
That is, $f'$ is constant on every element of the partition of Alice's
types given by 
\[ \{ \{(k_1, k_2, v_1, v_2)\}_{v_1,v_2\in\{0,1,2\}}
  \}_{(k_1,k_2)\in[K]\times[K]}. \]
We can thus apply 
\autoref{lem:noIrrelevantQuestions} 
to $\M_2$ to 
get a mechanism $\M_3$, which is still DSIC and has 
the same communication cost as $\M_2$,
and additionally never distinguishes between types of the form
$(k_1, k_2, \cdot, \cdot)$.
Formally, in $\M_3$, if we have 
$(k_1, k_2, v_1, v_2) \in \T_A(h)$ for some $v_1,v_2$,
then we have $(k_1, k_2, v_1, v_2) \in \T_A(h)$ for every $v_1,v_2$.

In principle, $\M_3$ can provides non-constant transfers to Alice.
Specifically, Bob can act at the root node in a way which changes the
transfer to Alice arbitrarily.
However, it turns out that this is all that is possible.
Namely, there cannot be any Alice node $h$ with two leaf nodes,
$\ell_1, \ell_2$ which are descendants of $h$,
such that Alice receives different transfers in $\ell_1$ and $\ell_2$,
If there were, Alice would have a strategic manipulation when her type is
indifferent on both sets of items, i.e. when $v_1 = v_2 = 2$.
Specifically, suppose the transfer at $\ell_1$ is higher than
the transfer at $\ell_2$. Because the game is perfect information,
there is some strategy of Bob such that, when Alice
follows actions directed towards $\ell_1$, Bob takes actions directed
towards $\ell_1$, and when Alice takes actions directed toward $\ell_2$,
Bob takes actions directed towards $\ell_2$.
Then, whichever Alice type takes $\ell_2$ has a strategic manipulation
against this strategy of Bob when her type has $v_1=v_2=2$.

Now, consider replacing every transfer to Alice in $\M_3$ with $0$
to get a mechanism $\M_4$.
By the reasoning in the preceding paragraph,
the strategic situation for Alice has not changed at any node.
More specifically, for each Alice node $h$, the transfers in $\M_3$ were
constant at every leaf below $h$.
This is still true in $\M_4$. Thus, if there were no strategic
manipulations in $\M_3$, there can be no strategic manipulations in $\M_4$.

Thus, $\M_4$ constitutes an implementation of $f'$ in which the transfers
to both agents are always $0$.
By simply ignoring the values of $v_1,v_2$, this constitutes
an implementation of $f$ from \autoref{sec:construction}.
By \autoref{thrm:EpicVsDsicNoTransfers}, this means that the communication
cost of $\M_4$ (and thus $\M_0$) is $\widetilde\Omega(2^m)$.

%
%

\end{proof}

%% file: FormalModel.tex
\section{Formal Definitions and Preliminary Analysis}
\label{app:FormalModel}

\paragraph{Environments and Implementations.}

An \emph{environment} for a set of $n$ \emph{players} $i=1,\ldots,n$
is a tuple $E = (Y, \T_1,\ldots, \T_n, u_1,\ldots, u_n)$.
Here, $Y$ is the set of \emph{outcomes}
and each $\T_i$ is the set of \emph{types} of player $i$.
Each $u_i : \T_i \times Y \to \R$ is the \emph{utility function}
of player $i$.
We say that a \emph{type} $t_i \in \T_i$ has \emph{utility}
$u_i(t_i, a) \in \R$ for an outcome $a\in Y$.
A \emph{social choice function} $f$ over $E$ is a mapping
$\T_1 \times\ldots\times \T_n \to Y$.
We restrict attention to deterministic social choice functions.

We say that a social choice function $f : \T_1 \times\ldots\times \T_n \to Y$
is \emph{incentive compatible} (or \emph{implementable}) (without transfers) if
for any $i$, $t_1\in\T_1,\ldots,t_n\in\T_n$, and $t_i'\in\T_i$, we have
\[ u_i(t_i, f(t_i,t_{-i}))
  \ge u_i(t_i, f(t_i',t_{-i})). \]
That is, each agent (weakly) maximizes their utility by reporting
their true type, regardless of the types of other agents.

Our paper works with the paradigm of
\emph{monetary transfers and quasilinear utilities}.
Unlike many prior papers, we make the distinction between environments with
transfers and without transfers explicit.
For any environment $E = (Y, \T_1,\ldots, \T_n, u_1,\ldots, u_n)$,
the corresponding \emph{quasilinear environment with transfers} is
$E' = (Y \times \R^n, \T_1,\ldots, \T_n, u_1',\ldots, u_n')$,
where $u_i'(t_i, (y, p_1, \ldots, p_n)) = u_i(t_i, y) + p_i$.
That is, the quasilinear environment adds \emph{transfers}
$p_1,\ldots, p_n$ to each agent, and the agents quasilinear
utility is the sum of its utility for the outcome and the transfer.
In this context, we call $u_i(t_i, a)$ the \emph{value} agent $i$ gets
(in order to distinguish it from agent $i$'s utility of
$u_i(t_i, a) + p_i$).
We say that a social choice function $f'$ over $E'$ \emph{computes}
a social choice function $f$ over $E$
if $f'$ satisfies $f'(t_1,\ldots,t_n) = ( f(t_1,\ldots,t_n) , p_1,\ldots, p_n)$
for each $(t_1,\ldots,t_n) \in \T_1,\ldots,\T_n$.
(That is, the function must agree on $Y$, but can be arbitrary
on the transfers.)

We treat transfers primarily as a tool for encouraging truthful behavior in
mechanisms\footnote{
  Alternative paradigms include studying ``budget balanced'' mechanisms or
  mechanisms that maximize revenue.
}.
Thus, a social choice function $f : \T_1\times\ldots\times \T_n \to Y$
over an environment $E$ is \emph{incentive compatible} (or
\emph{implementable}) (with transfers)
if there exists a function $f' : \T_1\times\ldots\times \T_n \to Y\times \R^n$
which computes $f$ and is incentive compatible in the 
quasilinear environment $E'$.
This is equivalent to the existence of $n$
transfer functions\footnote{ We do not make any assumptions
  (such as ``no positive transfers'' or ``individual rationality'')
  on the transfers the mechanism is allowed to use.
  This makes our imposibility results only stronger.
}
$p_1,\ldots,p_n : \T_1\times\ldots\times \T_n \to \R$
such that 
\[ u_i(t_i, f(t_i,t_{-i})) + p_i(t_i,t_{-i})
  \ge u_i(t_i, f(t_i',t_{-i})) + p_i(t_i',t_{-i}). \]
In this case, we say the transfer functions $(p_i)_{i=1,\ldots,n}$
\emph{incentivize} social choice function $f$.

A \emph{mechanism} $\M = (G, S_1,\ldots,S_n)$ over an environment
$E = (Y, \T_1,\ldots,\T_n, u_1,\ldots,u_n)$
consists of
\begin{itemize}
  \item An extensive form game $G$ for $n$ players
    with perfect recall and consequences in $Y$ (defined below).
  \item A type-strategy $S_i$ for each player $i$,
    which maps types $\T_i$ to (behavioural) strategies $s_i$
    of player $i$ in game $G$ (defined below).
\end{itemize}
Mechanism $\M$ over $E$
\emph{computes (without transfers)} a social choice
function $f$ over $E$ if we have $G( S_1(t_1),\ldots,S_n(t_n) ) = f(t_1,\ldots,t_n)$
for each profile of types $t_1,\ldots,t_n\in\T_1\times\ldots\times\T_n$.
A mechanism $\M$ over $E'$ \emph{computes (with transfers)} a social choice
$f$ over $E$ if $\M$ computes $f'$, for some $f'$ computing $f$.

We now present our incentive compatibility notions for mechanisms,
both with and without transfers
(recall that, if the mechanism has transfers, then
$u_i(t_i,(y,p_1,\ldots,p_n))$ denotes the quasilinear utility
$u_i(t_i, y) + p_i$).
An implementation $\M$ is \emph{ex-post Nash incentive compatible} (EPIC) if,
for any $i$, any types $t_1\in\T_1,\ldots,t_n\in\T_n$, and any behavioral strategy
$s_i'$ of player $i$, we have
(letting $S_{-i}(t_{-i}) = ( S_j(t_j) )_{j\ne i}$):
\[ u_i(t_i, G(S_i(t_i),S_{-i}(t_{-i})))
\ge u_i(t_i, G(s_i',S_{-i}(t_{-i}))). \]
An implementation is \emph{dominant strategy incentive compatible} (DSIC)
if, for any $i$, type $t_i$ of player $i$,
behavioral strategies $s_i'$ and $s_{-i}$ of all players,
\[ u_i(t_i, G(S_i(t_i),s_{-i}))
\ge u_i(t_i, G(s_i',s_{-i})). \]
Thus, ex-post Nash implementations are weaker, as they only require
$S_i(t_i)$ to be a best response when other agents are playing
strategies consistent with some $S_{-i}(t_{-i})$.

A \emph{direct revelation mechanism} is one in which
each agent is asked to simultaneously reveal their type to the mechanism,
and then the outcome is computed.
In such a mechanism, every possible strategy corresponds to some type,
and thus the mechanism is EPIC if and only if it is DSIC.

\paragraph{Extensive Form Games.}


A \emph{deterministic extensive form game with perfect recall
and consequences in $Y$} (hereafter called a \emph{game}) is a tuple
$G = (H, E, \mathcal{P}, A, \mathcal{A}, (\I_i)_{i\in [n]}, g)$ such that
\begin{itemize}
  \item $H$ is a set of states (also called nodes),
    and $E$ is a set of directed edges
    between the states, such that $(H, E)$ forms a finite directed tree
    (where every edge points away from the root).
    We denote typical elements of $H$ by $h$\footnote{
      This follows from the economics
      convention of identifying states with ``histories'', that is,
      the (unique) sequence of actions taken to arive in a certain node.
      We describe the game more concretely in terms of nodes of a tree
      because in some arguments we need to directly manipulate and change
      the game tree, which can alter these histories.
    }.
    Define $Z$ as the set of leafs, and let the root be called $h_0$.
    Moreover, define $\sigma_E : H \to 2^E$ such that
    $\sigma_E(h)$ is the set of edges leading out of state $h$ (that is,
    leading to successor nodes), and
    $\sigma_H : H\to 2^H$ such that $\sigma(h)$ is the set of states which
    are immediate successors of $h$ in the game tree.
    We write edges like $(h,h')\in E$, where $h$ is between the root and $h'$.
  \item $\mathcal{P} : H \to \{1, \ldots, n\}$ is the player
    choice function, which labels each non-leaf node in $H\setminus Z$ with
    the player who acts at that node.
    Define $H_i := \{ s \in S | \mathcal{P}(s)=i \}$ as the set of states where player
    $i$ is called to act.
    We assume that no player takes a consecutive turn,
    that is, for any $h$ and $h'\in \sigma_H(h)$, we have
    $\mathcal{P}(h)\ne\mathcal{P}(h')$\footnote{
      Note that this assumption is without loss of generality,
      because if player $i$ moves in two consecutive nodes $h$, $h'$, we
      could modify the game tree to put a trivial node (with a single
      action) of another player $j\ne i$ between $h$ and $h'$.
    }.
  \item $A$ is the set of actions, and $\mathcal{A} : E \to A$
    labels each edge with an action.
    $\mathcal{A}$ must be injective on each set $\sigma_E(h)\subseteq E$
    (that is, two edges below the same node cannot be labeled with
    the same action).
    Define $A : H \to 2^A$ 
    such that $A(h)$ is the set of
    actions available at state $h$
    (that is, $\{ \mathcal{A}(e) | e \in \sigma_E(h) \}$).
  \item For each player $i$, the set $\I_i$ (the \emph{information partition})
    is a partition of $H_i$ (the set of states where $i$ is called to act)
    such that for every set $I_i \in \I_i$,
    the set of actions available to $i$ are exactly the same at every set
    in $I_i$ (that is, $A(s) = A(s')$ for each $s, s'\in I_i$).
    The elements $I_i\in \I_i$ are called the \emph{information sets} of
    player $i$.
    Abusing notation slightly, we denote $A : \I_i \to A$ 
    such that $A(I_i) = A(s)$ for any $s\in N$.
    When a player acts in $G$, all that player knows
    is which information set they are in.
    The information sets must satisfy the following (known as the ``perfect
    recall'' assumptions\footnote{
      These assumptions just mean that the mechanism is not able to force agents
      to forget things they knew in the past of the game.
      While some of the game theory literature relaxes this assumption, we
      do not consider games without perfect recall here.
      This is in accordance with our adversarial model:
      we consider agents who know the already know the entire game tree in
      advance (although when they actually act in the mechanism, they only
      know what the mechanism tells them, that is, which information set
      they are in).
    }):
    \begin{itemize}
      \item Any path from the root to a leaf can cross a
        specific information set only once.
        That is, for any path $p = (h_0,h_1,\ldots,\allowbreak h_k)$ from the
        root to a leaf, we never have $h_i, h_j \in I_i \in \I_i$
        for $i\ne j$.
      \item If two nodes are in the same information set of player $i$,
        then $i$'s experience in reaching those nodes must be identical.
        More specifically, for node $h\in H_i$,
        we define $\psi_i(p)$, the \emph{experience of player $i$
        reaching $h$} as follows: take the (unique) path
        $p = (h_0, h_1, \ldots, h_k = h)$ from the root to $h$,
        and for each $j \in [k]$ at which $\mathcal{P}(h_j) = i$,
        write $(I_i^j, a^j)$ in order,
        where $I_i^j\ni h_j$ is the information set containing $h_j$
        and $a^j$ is the (unique) action which player $i$ takes at $h_j$
        to move the game to $h_{j+1}$.
        So $\psi_i(p)$ is an ordered, alternating list of information sets
        and actions $i$ takes at those information sets.
        We must have $\psi_i(h) = \psi_i(h')$ for any two histories in the
        same information set ($h, h'\in I_i \in \I_i$).
    \end{itemize}
  \item $g : Z \to Y$ labels each leaf node with an outcome from $Y$.
\end{itemize}


A \emph{(behavioral) strategy} $s_i$ of player $i$ is a function
$\I_i \to A$, such that $s_i(I_i) \in A(I_i)$ is an action available to player
$i$ at information set $I_i\in \I_i$.
The result of the mechanism under a behavioural strategy profile
$(s_1,\ldots,s_n)$ is the outcome in $Y$ in the leaf node which you arrive at
by iteratively following the action selected by each $s_i$.
We write this as $G(s_1,\ldots,s_n)$.
That is, if $(h_0, h_1, \ldots, h_k)$ is the path from the root
to a leaf such that each edge $(h_j, h_{j+1})$ is the unique edge
such that $\mathcal{A}(h_j, h_{j+1}) = s_i(I_i)$,
for $i = \mathcal{P}(h_j)$ and $h_j \in I_i \in \I_i$,
then we set $G(s_1,\ldots,s_n) = g(h_k)$.

A game is \emph{perfect information} if every information set is a
singleton. Observe that perfect information games cannot hide information
from players or allow more than one player to move simultaneously.

For a state $h\in H$, define the communication cost of $h$
$\lceil \log |A(h)| \rceil$, i.e. the number of bits needed for the agent
acting at $h$ to communicate their choice of action.
Define the communication cost of game $G$ as the maximum sum of the
communication costs of nodes on a path from the root to a leaf node
in $G$\footnote{
  Some prior works~\cite{FadelS09} limit mechanisms to at most two
  actions per node and define the communication cost as the maximum depth
  of the tree. This is equivalent to our definition up to constants, and
  our definition allows us to assume without loss of generality that no
  agent takes consecutive turns in the game.

  As is standard in the literature, we do not count
  the communication which the mechanism must send to the agents
  (to tell them which information set they are in).
  In perfect information games, one can argue that this is because the
  mechanism is simply run over a public communication channel.
  This does not apply in games of partial information.
  However, we note that
  counting the communication the mechanism would need to tell players their
  information set can increase the communication used
  by at most a quadratic factor. This is because at worst the mechanism
  needs to repeat to agent $i$ the messages of all agents which acted
  before the current $i$ node.
}.

Consider an environment $E = (Y, \T_1,\ldots, \T_n, u_1,\ldots, u_n)$
and an extensive form game $G$ with consequences in $Y$.
As we mentioned above, a \emph{type-strategy} $S_i$ is a mapping from $\T_i$ to 
behavioural strategies of player $i$ in game $G$.
Equivalently, a type-strategy is any function $S_i : \T_i \times \I_i \to A$
such that $s_i(t_i, I_i) \in A(I_i)$ for each $I_i\in\I_i$.
We let $S_i(t_i)$ denote the entire behavioural strategy.
For clarity, we capitalize type-strategies.
We typically refer to behavioural strategies simply as ``strategies'', and
specify explicitly when $S_i$ is a type-strategy.

\paragraph{Notation.}
When describing games between two players, Alice and Bob, we often use the
terms ``Alice node'' for a node where Alice acts.
We also denote such an $h$ with $\mathcal{P}(h)=A$.
Similarly, Bob nodes have $\mathcal{P}(h)=B$.

In auction-like domains, we typically identify the type with a valuation
function over the bundles of items received by a player.
For example, when the allowable types are some sets of functions from subsets
of $M$, and the allowable outcomes are partitions of the items to the $n$
players, we formally have $u_i(v_i, (A_1,\ldots,A_n)) = v_i(A_i)$.
Thus, we often write $v_i(A_i)$ in place of the entire outcome
$(A_1,\ldots,A_n)$.

For a player $i$ with type $t$, we say that $S_i(t)$ is the
``truth-telling'' strategy of player $i$, and the action
$S_i(t)(h)$ is the ``truth telling action'' at node $h$.
A strategy is a \emph{best response} to strategies $(s_j)_{j\ne i}$
for player $i$ with type $t$ if the strategy maximizes player $i$'s
utility across all possible strategies of player $i$.

\subsection{Describing the Game Tree via Sets of Types}
\label{sec:TreeTypes}

We now give some ways to regularize and describe EPIC mechanisms in
a natural way in terms of the types of agents.
Most of this language has been considered before (see
e.g.~\cite{FadelS09, BadeG17}).

Let $(G, S_1,\ldots,S_n)$ be a deterministic mechanism EPIC implementing $f$ over
environment $E = (Y, \T_1,\ldots,\allowbreak\T_n,u_1,\ldots,u_n)$.
For each node $h\in H$, let $\T(h) \subseteq \T_1\times\ldots\times\T_n$
denote the set of types $(t_1,\ldots,t_n)$ such that,
when agents play strategies $(S_1(t_1),\ldots,S_n(t_n))$, the computation
of $G(S_1(t_1),\ldots,S_n(t_n))$ enters state $h$
(that is, $h$ is on the path from the root to a leaf taken when computing
$G(S_1(t_1),\ldots,S_n(t_n))$).

\begin{lemma}
  \label{lem:typeSets}
  Each set $\T(h)$ is a rectangle. That is, $\T(h) =
  \T_1(h)\times\ldots\times\T_n(h)$ for some sets $\T_1(h),\ldots,\T_n(h)$.
  Moreover, if $\mathcal{P}(h)=i$
  then $\{\T_i(h')\}_{h'\in\sigma_H(h)}$ is a partition of $\T_i(h)$.
  If $\mathcal{P}(h)\ne i$
  then $\T_i(h') = \T_i(h)$ for all ${h'\in\sigma_H(h)}$.
\end{lemma}
\begin{proof}
  One can apply a standard rectangle argument to $\T(h)$.
  Specifically, suppose that $t = (t_1,\ldots,t_n),\allowbreak
  (t_1',\ldots,t_n') \in \T(h)$, and consider some player $i$ and
  type profile $(t_i', t_{-i})$.
  Every player other than $i$ takes the same action under $t$ and 
  $(t_i', t_{-i})$.
  Because the actions taken by each player along the path from the root to
  $h$ are unique, $S_i(t_i')$ must play the same actions along this path as
  does $S_i(t_i)$. Thus, $i$ will take the same action under $t_i$ and
  $t_i'$ at every node along the path where $i$ is called to act.
  So $(t_i', t_{-i})\in \T(h)$ as well.
  Applying this for each player $i=1,\ldots,n$ proves that $\T(h)$ is a
  rectangle. This makes $\T_1(h),\ldots,\T_n(h)$ well defined for each node
  $h$.

  If a player does not act at node $h$, then each successor of node $h$
  keeps the type set of that player the same, by definition.
  On the other hand, when player $i$ acts at $h$, every type in $\T_i(h)$
  takes exactly one action at $h$. This proves $\T_i(h')$ partitions
  $\T_i(h)$ as we let $h' \in \sigma_H(h)$ vary.
\end{proof}

Without loss of generality, we may assume 
that every node $h$ of the game has $\T(h)\ne \emptyset$\footnote{
  Removing nodes for which this is true can only decrease the communication
  complexity, while still computing the correct result in dominant
  strategies (if the original mechanism did so).
}.
Under this assumption, to specify an EPIC mechanism, it suffices to
specify a tree equipped with type sets $\T_i(h)$ satisfying the conclusions
of lemma~\ref{lem:typeSets}. In~\ref{app:simplifying}, 
we often describe modifying extensive form games in these terms.
%


%% file: RemovingAssumptions.tex
\section{Removing Additional Assumptions}
\label{app:simplifying}

%

The purpose of this appendix is to establish
two lemmas which restrict the structure of the mechanism we need to
consider. 
This first lemma apples in any two player (call them Alice and Bob)
social choice function where one player (say, Bob)
always receives the same utility in any outcome.
The two transformation are, intuitively,
1) telling Alice everything the mechanism knows (instead of using
information sets for Alice), and 2) deferring all questions which the
mechanisms asks Alice, but does not tell Bob, until the point in which the
mechanism actually does tell Bob.
For general mechanisms, both of these transformations may destroy the DSIC
property of the mechanism. However, they may only do so for \emph{Bob},
because both transformations only serve to enrich the strategy space of
Alice (because both transformations allow Alice to condition her actions on
more information revealed by Bob).
Luckily, because Bob can always guarantee himself the same utility
(in any mechanism correctly implementing the social choice function),
we can essentially ignore his incentives. On the other hand, for Alice, the
strategy space has only gotten richer, but if her strategy were truly
\emph{dominant} before, then it will remain dominant if she is allowed to
react to more information revealed by Bob.

\begin{lemma}
  \label{lem:perfectInformation}
  Suppose there is a social choice function $f$ for some two player environment.
  Call the players Alice and Bob, and consider some mechanism $\M$ DSIC
  implementing $f$ (with or without transfers) with communication cost $C$.
  Assume that when Bob plays his truth-telling strategy, he gets constant 
  utility $U$ in every outcome selected by $\M$. 
  Then there exists a perfect information mechanism $\M'$ DSIC
  implementing $f$ (with the same transfers as $\M$, if there are any)
  with communication cost at most $2C$.
\end{lemma}
\begin{proof}
  First, observe that for any such $\M$,
  truth-telling is always a dominant strategy for Bob.
  This is easiest to see as follows: $\M$ is, in particular, EPIC,
  so each node $h$ corresponds to a set of types $\T_A(h), \T_B(h)$ of Alice
  and Bob. Recall that without loss of generality, we assume that no type
  set is empty. Regardless of the strategy of Alice, if Bob acts according to
  his true type $t_B$, then the game will terminate in a leaf containing
  types $(t_A, t_B)$ for some Alice type $t_A$.
  The result will then be $f(t_A, t_B)$, which gets Bob utility $U$ by
  assumption. This is also the highest utility Bob can possibly achieve.


  Now, we describe a transformation from $\M_0 = \M$ to a perfect information
  mechanism $\M'$. This transformation consists of two major steps, the
  first eliminating information sets of Alice, and the second eliminating
  information sets for Bob.
  First, construct a new mechanism $\M_1$ which is identical to $\M_0$,
  except every information set of Alice is a singleton
  (i.e. Alice knows all information available to the mechanism).
  More precisely, simply replace Alice's information partition $\I_A$ with
  the trivial one, which places every node $h$ with $\mathcal{P}(h)=A$ into its own
  information set $\{h\}$.
  The communication cost of $\M_1$ is identical to that of $\M_0$.
  Moreover, every strategy of Alice $S_A(t_A)$ from $\M_0$ is still a valid
  strategy in $\M_1$, so $\M_1$ still correctly computes the social choice
  function. So we need only argue that the DSIC property is preserved.

  At a high level, we argue that shattering Alice's information sets
  preserves her dominant strategies, and thus can only harm
  the incentives of Bob (but this is a non-issue because Bob can always
  guarantee himself utility $U$).
  More formally, 
  suppose $\M_1$ were not DSIC. Because Bob always gets utility $U$
  under truth telling,
  Alice must be the player with a strategic manipulation.
  That is, there is a strategy
  $s_B$ of Bob in $\M_1$ such that
  $S_A(t_A)$ is not a best response in $\M_1$. Say that $s_A'$ gets higher
  utility for Alice when she has type $t_A$.
  Observe that $s_B$ also constitutes a strategy of Bob in $\M_0$, 
  because Bob's information sets are unchanged.
  Moreover, we can define a strategy $\widetilde s_A$ in $\M_0$ as follows:
  if $(h_0, h_1,\ldots,h_k)$ is the path taken 
  by $s_A', s_B$ in $\M_1$, and let
  $\widetilde s_A$ be any strategy in $\M_0$ such that $\widetilde s_A(I_A) =
  s_A'(h)$ for every $h$ in path $(h_0, h_1,\ldots,h_k)$
  (this is well-defined because any path from the root to a leaf must
  intersect an information set at most once, by the assumption of perfect recall).
  Then $\widetilde s_A, s_B$ achieves the same result in $\M_0$ as $s_A', s_B$
  does in $\M_1$.
  Then $\widetilde s_A$ is a valid strategy in $\M_0$, distinct from $S_A(t_A)$,
  which gives Alice with type $t_A$ higher utility
  that truth-telling when Bob plays $S_B$.
  Thus, $\M_0$ is not DSIC.
  So if $\M_0$ is DSIC, so is $\M_1$.

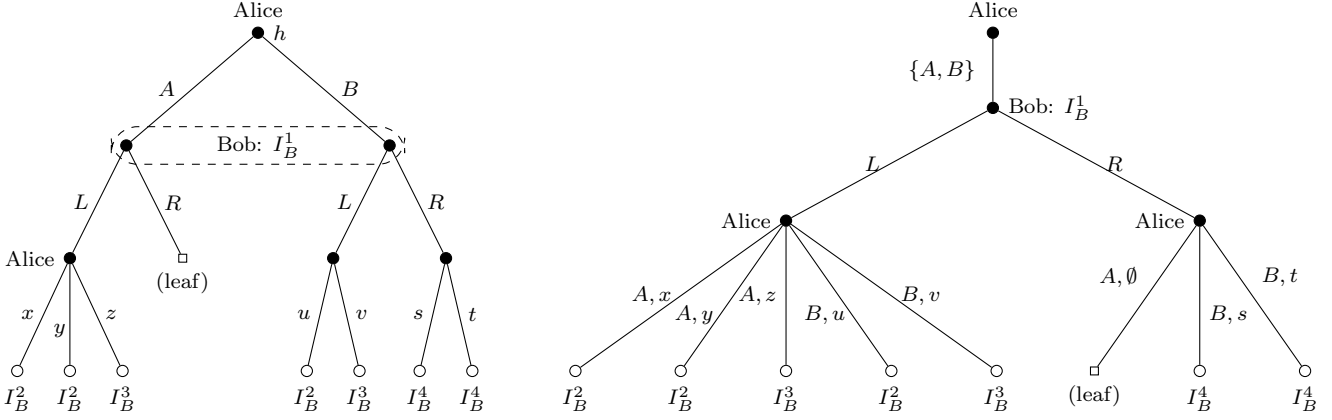
\begin{figure}
\tikzset{
  solid node/.style={circle,draw,inner sep=1.5,fill=black},
  hollow node/.style={circle,draw,inner sep=1.5},
  square node/.style={draw,inner sep=1.5}
}
\hspace{-7mm}
\begin{tikzpicture}[scale=1.0,font=\footnotesize]
  \tikzstyle{level 1}=[level distance=15mm,sibling distance=35mm]
  \tikzstyle{level 2}=[level distance=15mm,sibling distance=15mm]
  \tikzstyle{level 3}=[level distance=15mm,sibling distance=7mm]
\node(0)[solid node,label=above:{Alice},label=right:{$h$}]{}
  child{node(1)[solid node]{}
    child{node[solid node,label=left:{Alice}]{}
      child{node[hollow node, label=below:{$I_B^2$}]{} edge from parent node[left]{$x$}}
      child{node[hollow node, label=below:{$I_B^2$}]{} edge from parent node[below left, xshift=2]{$y$}}
      child{node[hollow node, label=below:{$I_B^3$}]{} edge from parent node[right]{$z$}}
      edge from parent node[left]{$L$}
    }
    child{node[square node,label=below:{(leaf)}]{} 
      edge from parent node[right]{$R$}
    }
    edge from parent node[left,xshift=-3]{$A$}
  }
  child{node(2)[solid node]{}
    child{node[solid node,label=below:{}]{}
      child{node[hollow node, label=below:{$I_B^2$}]{} edge from parent node[left]{$u$}}
      child{node[hollow node, label=below:{$I_B^3$}]{} edge from parent node[right]{$v$}}
      edge from parent node[left]{$L$}
    }
    child{node[solid node,label=below:{}]{}
      child{node[hollow node, label=below:{$I_B^4$}]{} edge from parent node[left]{$s$}}
      child{node[hollow node, label=below:{$I_B^4$}]{} edge from parent node[right]{$t$}}
      edge from parent node[right]{$R$}
    }
    edge from parent node[right,xshift=3]{$B$}
  };
  \draw[dashed,rounded corners=10]($(1) + (-.2,.25)$)rectangle($(2) +(.2,-.25)$);
  \node at ($(1)!.5!(2)$) {Bob: $I_B^1$};
\end{tikzpicture}
\qquad
\begin{tikzpicture}[scale=1.0,font=\footnotesize]
  \tikzstyle{level 1}=[level distance=10mm,sibling distance=35mm]
  \tikzstyle{level 2}=[level distance=15mm,sibling distance=55mm]
  \tikzstyle{level 3}=[level distance=20mm,sibling distance=14mm]
\node(0)[solid node,label=above:{Alice}]{}
  child{node(1)[solid node, label=right:{Bob: $I_B^1$}]{}
    child{node[solid node,label=left:{Alice}]{}
      child{node[hollow node, label=below:{$I_B^2$}]{} edge from parent node[left]{$A, x$}}
      child{node[hollow node, label=below:{$I_B^2$}]{} edge from parent node[below left, xshift=-4]{$A, y$}}
      child{node[hollow node, label=below:{$I_B^3$}]{} edge from parent node[left]{$A, z$}}
      child{node[hollow node, label=below:{$I_B^2$}]{} edge from parent node[below left, xshift=6]{$B, u$}}
      child{node[hollow node, label=below:{$I_B^3$}]{} edge from parent node[right]{$B, v$}}
      edge from parent node[left]{$L$}
    }
    child{node[solid node,label=left:{Alice}]{} 
      child{node[square node, label=below:{(leaf)}]{} edge from parent node[above left]{$A, \emptyset$}}
      child{node[hollow node, label=below:{$I_B^4$}]{} edge from parent node[below right]{$B, s$}}
      child{node[hollow node, label=below:{$I_B^4$}]{} edge from parent node[above right]{$B, t$}}
      edge from parent node[right]{$R$}
    }
    edge from parent node[left,xshift=-3]{$\{A, B\}$}
  } ;
\end{tikzpicture}
\caption{
  The node-by-node transformation needed to eliminate Bob's information sets.
  Specifically, this shows the transformation used for a single information
  set of Bob $I_B^1$ which occurs only in the subtree pictured.
  The hollow nodes in the bottom layer are Bob nodes,
  labeled with their information set, which may have
  large subtrees below them.
  The square node is a leaf.
  Note that the hollow subtrees are in a bijection before and after the
  transformation, and that there is a natural correspondence between
  strategies before and after the transformation.
  While Bob's available strategies are unchanged, Alice's strategies are
  now able to react to more information revealed by Bob.
  However, if her strategy is dominant in the original mechanism, it is
  still dominant after the transformation.
}
\label{fig:CondenseBobsInfoSets}
\end{figure}

  We now describe the second major step\footnote{
    As is common in the theory of extensive form games, the large amount of
    notation involved obscures the main points of this argument.
    See Figure~\ref{fig:CondenseBobsInfoSets} for the main intuition of
    the argument.
  }.
  This is composed of a long series of transformations, one for each node at
  which Alice acts, which proceed from the root of the tree downward.
  Specifically, let $\M_2^0 = \M_1$, and let $L$ denote the maximum number
  of nodes at which Alice acts along any path from the root to a leaf.
  We define
  \[ H_A^j = \{ h \in H_A\ \big|\ \text{there are $j-1$ ancestors $h'$ of $h$ 
    with $\mathcal{P}(h')=A$} \} \]
  That is, $H_A^j$ denotes the set of nodes $h$ at
  which Alice acts and for which there are $j-1$ nodes above $h$ in the
  tree at which Alice acts.
  For $j=1,\ldots,L$, we describe a transformation from $\M_2^{j-1}$ to $\M_2^j$.


  We prove by induction on $j$ that, after our transformation, all
  information sets of Bob are singletons. More formally, for each $j$, our
  inductive hypothesis is that for each $h\in H_A^j$,
  nodes, every ancestor of $h$ (Alice's or Bob's) is in an information set
  with only one element. Note that we must also prove that the game
  is still a valid mechanism (satisfying perfect recall) and still computes
  $f$ in dominant strategies.



  Assume by induction that all
  ancestors of nodes in $H_A^j$ have singleton information sets.
  For a node $h\in H_A^j$, we transforms $\M_2^{j-1}$
  to get a new mechanism $\M_2^{j,h}$.
  Let $A_{h,0}$ denote those actions in $A(h)$ which
  immediately result in a leaf node. Partition the remaining set of actions
  $A(h)\setminus \A_{h, 0}$ into
  $A_{h,1}, \ldots, A_{h,k}$ such that, for all Bob nodes $h'_1, h'_2 \in
  \sigma_H(h)$
  after $h$ in the game tree\footnote{
    Recall that we assume without loss of generality that players alternate
    turns along every path in the game tree.
  }, $h'_1$ and $h'_2$ are in the same information
  set if and only if actions $\mathcal{A}(h,h'_1)$ and
  $\mathcal{A}(h,h'_2)$ are in the same element of the partition.
  That is, partition Alice's actions according to Bob's
  information sets immediately below $h$.
  Let $U_{h, i}\subseteq \sigma_H(h)$ denote the set of Bob nodes
  $h'$ immediately following $h$ such that $\A(h,h')\in A_{h,i}$.
  For each set $A_{h,i}$, replace all of the Bob nodes in $U_{h,i}$ in
  $\M_2^{j-1}$ with a single node, labeled $U_{h,i}$, in $\M_2^{j, h}$.
  By our inductive hypothesis, no other node of Bob is in the same information set
  as these\footnote{
    To see this, use the perfect recall hypothesis. Any other Alice node
    $h'$ in $H_A^j$ must differ from Bob's point of view in some way:
    either it's a sibling of node $h$ (in which case Bob took a different action at the
    node preceding $h$) or there is a Bob node which differs along the path
    from the root to $h'$. By our inductive hypothesis, this different node
    is in a different information set.
    Thus, no descendant of $h'$ could be indistinguishable for Bob either.
  } $h'\in U_{h,i}$, so place $U_{h, i}$ in its own (trivial) information set 
  in $\M_2^{j,h}$.
  Replace all of Alice's actions at $h$ in $A_{h,i}$ with a single one labeled
  $A_{h,i}$\footnote{
    Note that this may create a node at which Alice has only a single
    action. This is fine and does not violate our assumptions,
    but you can remove such a node if you want (and then collapse the
    consecutive Bob nodes in the tree).
  }, which corresponds to an edge directed from $h$ to $U_{h, i}$,
  and modify $S_A$ so that all types which took any action in $A_{h,i}$
  now take action labeled $A_{h,i}$.
  Now, for each $A_{h,i}$, let $B_{h,i}$ denote the (unique) set of actions Bob
  had available in nodes in $U_{h,i}$ in $\M_2^{j-1}$.
  The new Bob node $U_{h, i}$ has exactly this set of actions $B_{h, i}$
  available in $\M_2^{j, h}$, with the same strategy profile $S_B$
  corresponding to them.
  Now, in $\M_2^{j,h}$ immediately below $U_{h,i}$,
  there is now a single Alice node corresponding to
  each action Bob may take in $B_{h,i}$.
  Call this node $h^A_{h, i, b}$ for each $b\in B_{h,i}$.
  We now construct the set of actions Alice has available at $h^A_{h,i,b}$,
  and the set of nodes below $h^A_{h,i,b}$.
  For any node $h''$ in $\M_2^{j-1}$ such that $(h, h')$ and $(h', h'')$
  are edges, where  $h'\in U_{h,i}$ and $\A(h', h'')=b$,
  we create actions at $h^A_{h,i,b}$.
  If $h''$ is a leaf node, we add an action to $h^A_{h,i,b}$
  labeled $(a, \emptyset)$, where $a=\A(h,h')$ is the action Alice takes leading to $h'$.
  Place the leaf node $h''$ below action $(a,\emptyset)$.
  If $h''$ is an Alice node, we add $|A(h'')|$ actions, one for each $a_2\in A(h'')$,
  labeled $(a_1,a_2)$, where $a_1=\A(h,h')$.
  For each such $a_2$, let $h'''$ be the node with $a_2=\A(h'',h''')$,
  and place $h'''$ below $h^A_{h,i,b}$ corresponding to $(a_1,a_2)$.
  
  By construction, nodes ``outside of the transformation'' below $h$
  are now in a bijection in $\M_2^j$ and $\M_2^{j, h}$.
  Specifically, a leaf node along path corresponding to $(a_1, b_1)$, where
  $a_1 \in A_{h,i}$ and $b_1\in B_{h,i}$, now corresponds to action
  sequence $(A_{h,i}, b_1, (a_1, \emptyset) )$.
  Moreover, a node along path $(a_1, b_1, a_2)$, where $a_1 \in A_{h,i}$,
  $b_1\in B_{h,i}$, and $a_2\in A(h'')$ (for whichever $h''$ corresponds
  to following $(a_1, b_1)$)
  now corresponds to sequence of actions $( A_{h,i}, b_1, (a_1, a_2) )$.
  Alice's strategy $S_A$ continues to take actions in a natural way:
  any type which took action $a_1$ at $h$ and $a_2$ at $h''$ will take
  action $A_{h,i}\ni a_1$ at $h$ and $(a_1, a_2)$ at $h''$.
  In the new mechanism $\M_2^{j,h}$, we place the entire subtree from the old mechanism
  $\M_2^{j-1}$ below the corresponding newly created node.
  This completes our description of the transformation at node $h$.
  Observe that all Bob nodes in $\sigma_H(h)$ directly below $h$ in 
  $\M_2^{j,h}$ are now in trivial information sets.

  We must prove that our transformation preserves the correctness of the
  mechanism.
  Because the nodes outside the transformation are in a bijection, and
  strategies $S_A, S_B$ still result in entering the correct subtree
  outside of the transformed nodes, $\M_2^{j,h}$ still implements $f$
  (with the same transfers as $\M_2^{j-1}$, if there are any).
  Moreover, the game still satisfies perfect recall, as we have only
  reduced the information sets used by the mechanism.
  Thus, we need only argue that incentives are preserved.
  That is, we need to formally show that if $\M_2^{j-1}$ is DSIC for Alice, so is
  $\M_2^{j,h}$. 
  At a high level, this holds for the same reason that the DSIC property was
  preserved by our first transformation.
  Namely, we have only enriched the strategy space of Alice, while Bob's
  strategy space is unchanged -- thus, if Alice's strategy were already
  dominant, than it will still be dominant after the transformation.

  In detail, suppose that $\M_2^{j-1}$ is DSIC, yet $\M_2^{j,h}$ is not.
  It must be Alice who has a strategic manipulation, because Bob has a
  dominant strategy for any mechanism implementing $f$.
  Thus, there must be a strategy
  $s_B$ of Bob in $\M_2^{j,h}$ such that
  $S_A(t_A)$ is not a best response of Alice $\M_2^{j,h}$.
  Say that $s_A'$ gets higher utility for Alice when she has type $t_A$.
  Observe that $s_B$ also constitutes a strategy of Bob in $\M_2^{j-1}$,
  because Bob's information sets are in a bijection in $\M_2^{j-1}$ and $\M_2^{j,h}$.
  Moreover, we can define a strategy $\widetilde s_A$ in $\M_2^{j-1}$ which
  corresponds to $s_A'$ in a natural way:
  $\widetilde s_A$ exactly corresponds to $s_A'$ outside the transformed area,
  and inside the transformed area, if $s_A'$ takes action $A_{h,i}$ at $h$ and
  $(a_1, a_2)$ at $h''$ (for whichever $h''$ is on the path selected by
  Bob under $s_B$), then Alice takes action $a_1$ at $h$ and $a_2$ at $h''$
  in $\M_2^{j-1}$ (and the actions taken by $\widetilde{s_A}$ at other possible $h''$
  is selected arbitrarily).
  Then $\widetilde s_A, s_B$ achieves the same result in $\M_2^{j}$ as $s_A', s_B$
  does in $\M_2^{j,h}$, and thus Alice has a strategic manipulation in $\M_2^{j-1}$.
  Thus, $\M_2^{j-1}$ is not DSIC. So if $\M_2^{j-1}$ is DSIC, so is $\M_2^{j,h}$.

  We have shown that the above transformation correctly eliminates Bob's
  information sets in $\sigma_H(h)$ for a single Alice node $h\in H_A^j$.
  Observe that for distinct Alice nodes $h$ in layer $j$, the transformation
  from $\M_2^{j-1}$ to $\M_2^{j,h}$ is independent.
  Thus, applying the transformation in sequence for all Alice nodes in
  $H_A^j$ eliminates information sets in the layer of Bob nodes directly
  below $H_A^j$. 
  Construct $\M_2^j$ by applying the transformation on top of $\M_2^{j-1}$
  for all $h\in H_A^j$.
  By induction on $j$, we know that $\M_2^j$
  eliminates all information sets from node which are ancestors of
  $H_A^{j+1}$.
  Applying this for all layers $j=1,\ldots,L$ eliminates all information
  sets. Thus, define $\M' = \M_2^L$ to be the result of applying the
  transformation across all layers of Alice nodes.

  Finally, we bound the communication cost of $\M'$.
  For each Alice node $h$, at worst, the communication cost of $h$ is
  added to the Alice nodes two layers below $h$. So along any path the cost of
  each node is incurred at most twice, and all told, the cost of $\M$ is at most
  doubled.



\end{proof}

\begin{remark}
  If the above transformations are applied to a two player mechanism in
  which Bob does not have constrained incentives, the result is a perfect
  information mechanism with the following properties:
  the mechanism is DSIC for Alice, but only EPIC for Bob.
  That is, truth-telling is a best response for Alice for any strategy Bob
  plays, but truth-telling is only guaranteed to be a best response of Bob
  when Alice plays a strategy consistent with $S_A(t_A)$ for some type
  $t_A$.
\end{remark}

This second lemma captures an intuitive property:
if there is information irrelevant to the outcome of the social choice function,
then mechanisms need not ask about that information (neither to compute the
function, nor to guarantee incentive properties).
This property feels especially necessary for games of perfect information.
Intuitively, we'd like to say that, if the mechanism ever asked a question that didn't
matter, then the mechanism could not hope to be DSIC, because the other
players could react maliciously to that irrelevant question in order to give the player
a worse outcome. This argument is not literally true, because the
mechanism may ask such a question at a point where the question has already
become irrelevant, for example, if the outcome is already
determined\footnote{
  If the mechanism is an implementation with transfers, then
  it is also possible that there are two Alice types $t, t'$ which do not
  matter from the point of view of the social choice function $f$,
  but which effect the transfers given to Bob.
  This is why we need to again assume that Bob gets constant utility in the
  next lemma.
}.
However, it is true that we can assume this property holds without loss of generality
in the games we are interested in.
We describe and prove this fact formally
using the language of Section~\ref{sec:TreeTypes}.

\begin{lemma}
  \label{lem:noIrrelevantQuestions}
  Consider any perfect information DSIC mechanism $\M$ for two players, Alice and
  Bob, implementing social choice function $f$ (with or without transfers).
  Suppose that in every outcome selected by $\M$, Bob
  receives zero transfers, and when Bob plays his truth-telling strategy,
  he gets constant utility $U$.
  Let $A_1,\ldots,A_k$ be any
  partition of the types of Alice such that $f$ does not depend on the
  difference between elements within that partition
  (that is, for any Alice types $a, a' \in A_i$, for all Bob types $b$,
  we have $f(a, b) = f(a', b)$).
  Then there exists another DSIC mechanism $\M'$ which implements $f$,
  whose communication cost is at most that of $\M$, which never
  asks Alice to differentiate between types in $A_i$.
  That is, for each leaf $\ell$ of $\M'$, if $\T_A(\ell) \cap A_i \ne \emptyset$,
  then $A_i \subseteq \T_A(\ell)$.
\end{lemma}

\begin{proof}

  For each $A_i$, pick an arbitrary representative $a_i\in A_i$.
  Simply create a new mechanism $\M'$ where all types in
  $A_i$ take whichever action is take by $a_i$, instead of whichever
  action they took before.
  More precisely, create a new mechanism $\M'$, over the social choice
  environment restricted to where Alice's types are $\{a_1,\ldots,a_k\}$,
  which corresponds with $\M$ on these types. Remove any nodes $h$ whose type
  sets $\T(h)$ are empty.
  Now, let this define a mechanism $\M'$ over the original environment in
  the natural way, specifically, each Alice type $a$ follows that action
  take by $a_i$ for whichever $i$ has $a\in A_i$.


  This creates $\M'$ with no increase in communication cost.
  By definition, $\M'$ still implements $f$.
  Thus, we just need to reason that incentive properties are preserved.
  Note that Bob always has a dominant strategy in $\M'$,
  as he still receives zero transfers and value $U$ for each outcome selected.
  If $\M$ is an implementation with transfers, then from a standard
  argument\footnote{
    To see this, recall that $\M$ is in particular EPIC with transfers,
    and thus implements $(f, p_A, p_B)$ for some transfers which satisfy
    $u_A(t_A, f(t_A, t_B)) + p_A(t_A, t_B)
      \ge u_A(t_A, f(t_A', t_B)) + p_A(t_A', t_B)$
    for all $t_A, t_B, t_A'$.
    If there were types $t_A, t_A', t_B$ such that
    $f(t_A, t_B) = f(t_A', t_B)$ yet $p_A(t_A) < p_A(t_A')$, then when
    Alice has type $t_A$, she could benefit by deviating to $t_A'$ when
    Bob has type $t_B$.
    Thus, the transfers to Alice are a function only of Bob's type and
    the outcome selected.
  } it follows that the transfers
  must be a function only of Bob's type and the outcome selected.
  Thus, $\M'$ and $\M$ give the same transfers to Alice\footnote{
    For mechanisms in which Bob receives nonzero transfers, $\M'$ could
    in principle give different transfers to Bob than $\M$.
  }.

  Suppose there is a strategic manipulation in $\M'$.
  In particular, in $\M'$ there must be a type for Alice $t_A\in A_i$,
  some strategy $s_B$ of Bob, and strategy $s_A'$ of Alice which dominates
  following $S_A(t_A)$.
  Consider the computation path under $(s_A', s_B)$ and $(S_A(t_A), s_B)$.
  Let $h$ denote the first node along these paths where
  strategic manipulation $s_A'$ takes a different action from $S_A(t_A)$
  (and note that $h$ must be an Alice node).
  There must be some Bob-types $t_B^g, t_B^b \in \T^{\M'}_B(h)$\footnote{
    For clarity, in this argument we let $\T^{\M}(h)$ denote the type-set at
    node $h$ in the original mechanism $\M$, and we let $\T^{\M'}(h)$ denote
    the type set in $\M'$ (whenever the node corresponding to $h$ still exists
    in $\M'$).
  } such that,
  when play proceeds according to $(s_A', s_B)$, the leaf node computed contains Bob type
  $t_B^g$, and when $(S_A(t_A), s_B)$ is played, the leaf contains $t_B^b$
  (intuitively, $t_B^g$ is the ``good type'' for Alice which Bob can
  pretend to be, and $t_B^b$ is the ``bad type'').
  Recall that $h$ is also a node in the original mechanism $\M$.
  
  If $t_A \in \T_A^{\M}(h)$, then it is easy to see that $\M$ also has a
  strategic manipulation
  ($s_B$ and $s_A'$ can both be extended to full strategies in $\M$,
  and $S_A(t_A)$ is still dominated by $s_A'$ when Alice has type $t_A$
  and Bob plays $s_B$).
  
  Otherwise, $t_A$ got moved into $\T_A^{\M'}(h)$ to be in the same node
  as some $t_A'$ such that $t_A, t_A' \in A_i$.
  Consider the paths taken by computing $(S_A(t_A), s_B), (S_A(t_A'), s_B)$
  in $\M$. Let $h^*$ denote the last point along both of these paths from the root to $h$ in
  which $t_A$ and $t_A'$ are in the same node. The Bob-types in this node are a
  superset of those in $h$, and thus, both $t_B^g$ and $t_B^b$ are in $h^*$.
  Moreover, because the game is perfect
  information, we can construct strategy $s_B^*$ of Bob in $\M$ as follows:
  if Alice follows the action taken by $t_A$ at $h^*$, act as if Bob has
  type $t_B^b$, but if Alice continues at each successive node to play the
  action played by $t_A'$, then Bob acts as if he has type $t_B^g$.
  When Bob plays $s_B^*$ and Alice plays $S_A(t_A)$, she gets whichever ``bad outcome'' she got
  under truthful play in $\M'$ (and she receives the same transfer as well).
  When Bob plays $s_B^*$ and Alice deviates and plays according to $S_A(t_A')$,
  she'll get the good outcome from $\M'$.
  Thus, truth-telling is not a best response of Alice in $\M$ 
  and $\M$ is not DSIC either.
  
  So if $\M$ is DSIC, then $\M'$ is DSIC as well.

\end{proof}